\documentclass[letterpaper,aps,pra,notitlepage,nofootinbib,longbibliography,twocolumn,10pt,superscriptaddress,floatfix]{revtex4-1}

\usepackage{microtype}
\usepackage{amsmath}
\usepackage{amssymb}
\usepackage{amsfonts}
\usepackage{amsthm}
\usepackage{youngtab}
\usepackage{graphicx}

\usepackage{color}
\definecolor{darkblue}{rgb}{0.,0.,0.4}
\definecolor{darkred}{rgb}{0.5,0.,0.}
\usepackage[pdftex,colorlinks=true,linkcolor=darkblue,citecolor=darkred,urlcolor=blue]{hyperref}

\newtheorem{thm}{Theorem}
\newtheorem*{thm*}{Theorem}

\newtheorem{lem}[thm]{Lemma}

\theoremstyle{definition}

\theoremstyle{plain}

\newcommand{\ket}[1]{\left| {#1} \right\rangle }
\newcommand{\bra}[1]{\left\langle  {#1} \right| }
\newcommand{\sgn}{\mathop{\mathrm{sgn}}\nolimits}
\newcommand{\rd}{\mathrm{d}}
\newcommand{\proj}[1]{\left|#1\right\rangle\left\langle #1\right|}

\newcommand{\nc}{\newcommand}
\nc{\rnc}{\renewcommand}

\nc\eps{\epsilon}

\nc\bbC{\mathbb{C}}
\DeclareMathOperator*{\E}{\mathbb{E}}

\nc\bbF{\mathbb{F}}
\nc\bbM{\mathbb{M}}
\nc\bbN{\mathbb{N}}
\nc\bbR{\mathbb{R}}
\nc\bbS{\mathbb{S}}
\nc\bbZ{\mathbb{Z}}

\nc\bp{\mathbf{p}}
\nc\bq{\mathbf{q}}

\nc\benum{\begin{enumerate}}
\nc\eenum{\end{enumerate}}
\nc\bit{\begin{itemize}}
\nc\eit{\end{itemize}}

\newcommand{\lemref}[1]{Lemma~\ref{lem:#1}}

\nc{\todo}[1]{\textcolor{red}{todo: #1}}

\nc\cA{\mathcal{A}}
\nc\cB{\mathcal{B}}
\nc\cC{\mathcal{C}}
\nc\cD{\mathcal{D}}
\nc\cE{\mathcal{E}}
\nc\cF{\mathcal{F}}
\nc\cG{\mathcal{G}}
\nc\cH{\mathcal{H}}
\nc\cI{\mathcal{I}}
\nc\cJ{\mathcal{J}}
\nc\cK{\mathcal{K}}
\nc\cL{\mathcal{L}}
\nc\cM{\mathcal{M}}
\nc\cN{\mathcal{N}}
\nc\cO{\mathcal{O}}
\nc\cP{\mathcal{P}}
\nc\cQ{\mathcal{Q}}
\nc\cR{\mathcal{R}}
\nc\cS{\mathcal{S}}
\nc\cT{\mathcal{T}}
\nc\cU{\mathcal{U}}
\nc\cV{\mathcal{V}}
\nc\cW{\mathcal{W}}
\nc\cX{\mathcal{X}}
\nc\cY{\mathcal{Y}}
\nc\cZ{\mathcal{Z}}

\DeclareMathOperator{\GL}{GL}
\DeclareMathOperator{\tr}{tr}
\DeclareMathOperator{\id}{id}
\DeclareMathOperator{\poly}{poly}
\DeclareMathOperator{\Var}{Var}

\def\be#1\ee{\begin{equation}#1\end{equation}}
\def\bea#1\eea{\begin{eqnarray}#1\end{eqnarray}}
\def\beas#1\eeas{\begin{eqnarray*}#1\end{eqnarray*}}
\def\ba#1\ea{\begin{align}#1\end{align}}
\def\bas#1\eas{\begin{align*}#1\end{align*}}
\def\bpm#1\epm{\begin{pmatrix}#1\end{pmatrix}}

\def\eq#1{(\ref{eq:#1})}

\rnc\L{\left}
\nc\R{\right}
\nc\ra{\rightarrow}
\nc\ot{\otimes}

\begin{document}

\title{Sample-optimal tomography of quantum states}

\author{Jeongwan Haah}
\affiliation{Station Q Quantum Architectures and Computation, Microsoft Research, Redmond, Washington, USA}
\affiliation{Center for Theoretical Physics, Massachusetts Institute of Technology, Cambridge, Massachusetts, USA}
\author{Aram W. Harrow}
\affiliation{Center for Theoretical Physics, Massachusetts Institute of Technology, Cambridge, Massachusetts, USA}
\author{Zhengfeng Ji}
\affiliation{Centre for Quantum Computation \& Intelligent Systems, Faculty of Engineering and Information Technology, University of Technology, Sydney, NSW 2007, Australia}
\affiliation{Institute for Quantum Computing, University of Waterloo, Waterloo, Ontario, Canada}
\affiliation{State Key Laboratory of Computer Science, Institute of Software, Chinese Academy of Sciences, Beijing, China.}
\author{Xiaodi Wu}
%\affiliation{Center for Theoretical Physics, Massachusetts Institute of Technology, Cambridge, Massachusetts, USA}
\affiliation{Department of Computer and Information Science, University of Oregon, Eugene, Oregon, USA}
\author{Nengkun Yu}
\affiliation{Institute for Quantum Computing, University of Waterloo, Waterloo, Ontario, Canada}
\affiliation{Centre for Quantum Computation \& Intelligent Systems, Faculty of Engineering and Information Technology, University of Technology, Sydney, NSW 2007, Australia}
\affiliation{Department of Mathematics \& Statistics, University of Guelph, Guelph, Ontario, Canada}

%\date{\today}

\begin{abstract}
  It is a fundamental problem to decide how many copies of an unknown
  mixed quantum state are necessary and sufficient to determine the
  state.  Previously, it was known only that estimating states to
  error $\epsilon$ in trace distance required $O(dr^2/\epsilon^2)$
  copies for a $d$-dimensional density matrix of rank $r$.  Here, we
  give a theoretical measurement scheme (POVM) that requires
  $O (dr/ \delta ) \ln (d/\delta) $ copies of $\rho$ to error
  $\delta$ in infidelity, and a matching lower bound up to logarithmic
  factors.  This implies $O( (dr / \epsilon^2) \ln (d/\epsilon) )$ copies
  suffice to achieve error $\epsilon$ in trace distance.
We also prove that for independent (product) measurements,
$\Omega(dr^2/\delta^2) / \ln(1/\delta)$ copies are necessary
in order to achieve error $\delta$ in infidelity.
  For fixed $d$, our measurement can be implemented on a quantum computer in
time polynomial in $n$.  \preprint{MIT-CTP/4699}
\end{abstract}

\maketitle

\begin{table*}[bt]
\centering
\caption{
 Conditions for the quantum state tomography with high success
 probability.
 $\delta$ denotes the accuracy goal measured in the infidelity $1-F(\rho,\hat \rho)= 1- \|\sqrt{\rho}\sqrt{\hat\rho}\|_1$,
 and $\epsilon$ denotes that in the trace distance $T(\rho,\hat \rho) = \frac 12 \| \rho - \hat \rho \|_1$.
 The upper bound in terms of the infidelity implies that in terms of
 trace distance; $n \le O(d^2/\epsilon^2) \ln (d/\epsilon)$.
The lower bound in terms of the trace distance implies that in terms
of infidelity; e.g. $n \ge \Omega( d^2 / \delta )$.
The lower bound for the independent measurements in rank $r$ case
implies $n \ge \Omega( dr^2 / \epsilon^2 \ln(1/\epsilon))$.
The previously known upper bound on $n$ already used only independent
measurements; thus our lower bounds show that this result was
essentially optimal.}
\begin{tabular}{|c|c|c|c|}
\hline
& \multicolumn{2}{|c|}{Our result} & Previous result\\
\hline
 & for general $\rho \in \mathbb C^{d \times d}$ & \multicolumn{2}{|c|}{for $\rho$ of rank at most $r$ }\\
\hline
Sufficient & $n \le O(d^2/\delta) \ln (d/\delta)$ & $n \le O(rd/\delta) \ln (d/\delta)$ & $n \le O(r^2 d / \epsilon^2)$ \cite{KRT14} See Sec.~\ref{sec:KRT14}.\\
\hline
Necessary & $n \ge \Omega\left( d^2/\epsilon^2 \right)$ & $n \ge \Omega\left( rd/\epsilon^2 \right) / \ln( d / r \epsilon)$
  & $n \ge \Omega(1/\epsilon^2) + \tilde \Omega(rd)$ \cite{FlammiaLiu2011Direct}\\
\hline
\parbox{20ex}{\vspace{1mm} Necessary using \\independent\\
  measurements \vspace{1mm} }
& $n\geq \Omega(d^3/\epsilon^2)$
& $n\geq \Omega(dr^2/\delta^2\ln(1/\delta))$
& $n \geq \Omega(1/\delta^2\ln(1/\delta))$ See Sec.~\ref{sec:prior}.\\
\hline
\end{tabular}
\label{tb:result}
\end{table*}

Given $n$ copies of an unknown $d$-dimensional quantum state $\rho$,
how accurately can $\rho$ be estimated? This fundamental question
arises both in quantum information theory and in the interpretation of
experimental results.  Since $\rho$ has $d^2-1$ real parameters, it is
reasonable to conjecture that $\Theta(d^2)$ measurements are necessary
and sufficient to estimate $\rho$ to constant accuracy. On the other
hand, even distinguishing a fair coin from a coin biased to obtain
heads with probability $1/2+\eps$ requires $\Omega(1/\eps^2)$
measurements.

In this paper we show that the number
of copies required to estimate $\rho$ with precision $\eps$ scales
roughly with both $d^2$ and $1/\eps^2$.
More precisely, if the fidelity goal is $1-\delta$,
we prove an $\Omega(d^2/ \delta)$ lower bound
and an $O((d^2/ \delta) \ln(d/\delta))$ upper bound
on the number of required copies.
When the state $\rho$ is guaranteed to have rank $\leq r$
we show an $O((dr/ \delta)\ln(d/\delta))$ upper bound
and an $\Omega((dr/\delta)/\ln(d/r\delta))$ lower bound.
We also prove a lower bound $\Omega(dr^2/\delta^2)/\ln(1/\delta)$
for independent measurement schemes
where individual copies are measured independently and then
the outcomes are processed to output an estimate $\hat \rho$.
Our result is summarized in Table~\ref{tb:result}.

\paragraph*{Notation}
We use the convention that $\Omega(x)$ means a
function that is asymptotically $\geq c_1x$ for a constant $c_1>0$,
$O(x)$ means $\leq c_2 x$ for a constant $c_2>0$ and $\Theta(x)$ means
both $O(x)$ and $\Omega(x)$.
Notation $\tilde O()$ means that we neglect $\ln$ factors.
$\ln$ and $\exp$ are base-$e$.

\section{Accuracy measures}

The fidelity of two quantum states $\rho,\sigma$ is
$F(\rho,\sigma) := \tr\sqrt{\sqrt\rho ~\sigma \sqrt\rho}$,
the ``infidelity'' is $1-F$, represented by $\delta$, and
their trace distance is $T(\rho,\sigma):=\frac 1 2 \|\rho-\sigma\|_1$, represented by $\epsilon$.
These are related by~\cite{FuchsGraaf2007Cryptographic}
\begin{align}
1- F \leq T \leq \sqrt{1-F^2}.
\label{eq:FuchsGraaf}
\end{align}

We derive an upper bound in terms of fidelity and a lower bound in
terms of trace distance, in each case implying a near-optimal bound in
terms of the other quantity.   Here we discuss why fidelity is in many
ways a natural quantity for tomography~\cite{Wootters1981}.
Tomography is essentially a state discrimination procedure
where one distinguishes $\rho^{\otimes n}$ from $\sigma^{\otimes n}$.
The statistical distinguishability of these states is measured by the trace distance
$T_n = T(\rho^{\otimes n}, \sigma^{\otimes n})$,
which is in general much larger than $T(\rho, \sigma)$;
this amplification is what enables the tomography.
The asymptotic behavior of $T_n$ can be quantified as
\[
 \frac 1 2  F(\rho, \sigma)^{2n} \leq 1- T_n \leq F(\rho, \sigma)^{n}
\]
by Eq.~\eqref{eq:FuchsGraaf} and $F(\rho^{\otimes n} , \sigma^{\otimes n}) = F(\rho,\sigma)^n$.
This means that
$\ln (1/F)$ or infidelity gives nearly sharp bounds on the rate
at which $T_n$ converges to 1;
the actual rate%
\footnote{
The exact scaling of $1- T_n$ for large $n$ is known to be $C^n$
where $C = C(\rho,\sigma) = \inf_{ 0 \leq s \leq 1 } \tr( \rho^s
\sigma^{1-s} )$, and $\ln(1/C)$ is called the quantum Chernoff distance%
~\cite{NussbaumSzkola2009Chernoff, AudenaertCalsamigliaMuEtAl2007Chernoff}.
}
is between $\ln(1/F)$ and $2 \ln (1/F)$.
In particular, for fixed $d$,
the state discrimination is possible to infidelity $\delta$
using $n = \Theta( 1/ \delta )$ copies.
Our upper bound on $n$ in terms of fidelity proves
that the POVM we present in this paper
indeed accomplishes the discrimination task using $n =\tilde O( 1/ \delta )$ copies.
On the contrary, the corollary upper bound in terms of trace distance
sometimes over-estimates the sufficient number of samples by an
unbounded amount.
As a simple example, consider qubit states
\[
 \rho = \begin{pmatrix} 1 & 0 \\ 0 & 0 \end{pmatrix} \quad \text{and} \quad
 \sigma = \begin{pmatrix} 1- \eps & 0 \\ 0 & \eps \end{pmatrix} ,
\]
between which the trace distance is $\eps$ and the infidelity is $1-\sqrt{1-\epsilon} \simeq \eps/2$.
The trace distance bound only says $n = \tilde O( 1/ \eps^2 )$ copies
are sufficient to distinguish them,
whereas the fidelity bound says $n = \tilde O( 1/ \eps )$ copies are sufficient.

\section{Previous Results}
\label{sec:prior}

Quantum state estimation has been extensively studied, going back at
least to the work of Helstrom~\cite{Helstrom69},
Holevo~\cite{holevo-book} and others from around 1970.
Many of the rigorous results are for the special cases when $d=2$ or
$r=1$, or give an uncontrolled or suboptimal $d$ dependence (e.g.~with
$n$ scaling as $f(d)/\delta$ for unknown $f$) or discuss related
problems such as spectrum estimation, parameter estimation or
determining the identity of a state drawn from a discrete set.
In this paper we will consider optimal measurements
(also called ``collective'' measurements)
and will not discuss the extensive literature
on independent or adaptive measurements.

For $d=2$ (i.e.~qubits), the optimal infidelity was shown in
\cite{BBMR04,BBGMM06,GK06,GJK08,HayashiM08} to scale as $1/n$.
This scaling was generalized to qudits in \cite{GK09} (see also Section 6.4 of
\cite{hayashi-book}), but with an uncontrolled dependence on $d$
(i.e.$n$ scales as $f(d)/\delta$ for unknown $f(\cdot)$); see also \cite{Keyl06}.
In many settings (e.g.~minimax estimation) one can show that covariant
measurements are optimal.  If one further assumes that $\rho$ is pure
then the optimal estimation strategy has a simple form and $n$ should
scale as $\Theta(d/\delta)$~\cite{Hayashi98,holevo-book}; see also
\cite{Chiri10} where further connections were made to cloning and de
Finetti theorems.

Another major theme in recent work has been the study of various forms
of restricted measurements, e.g.~independent measurements with a
limited number of measurement settings.  Intermediate between
independent measurements and unrestricted (also called ``collective''
or ``entangled'') measurements are {\em adaptive} measurements in
which the copies of $\rho$ are measured individually, but the choice
of measurement basis can change in response to earlier measurements.

On the achievability side for independent measurements, a sequence of
works~\cite{compressed,FlammiaGrossLiuEtAl2012,Vlad13,KRT14} showed that
$n=O(dr^2/\eps^2)$ copies are sufficient to obtain trace distance
$\leq \eps$ with high probability.%
\footnote{%
  The earlier
  papers~\cite{compressed,FlammiaGrossLiuEtAl2012} achieved
  $n=\tilde O(d^2r^2/\eps^2)$.
  The improved   $n=O(dr^2/\eps^2)$ performance
  is achieved by analyzing Theorem~2 of \cite{KRT14}.  This is not
obvious from their theorem statement, but we explain the connection in
Sec.~\ref{sec:KRT14}.
}
On the other hand, even for $d=2$, adaptive and collective
measurements are known to have asymptotically better error scaling, at least when
measured in terms of infidelity. The usual intuition is that $n$
should scale as $1/\delta^2$ for independent measurements and
$1/\delta$ for adaptive or collective measurements; e.g.~see
\cite{MahlerRozemaDarabiEtAl2013Adaptive} for numerical evidence.
Refs.~\cite{BBGMM06,HayashiM08} showed that adaptive measurements
could achieve $n = O(1/\delta)$ scaling.
When a POVM contains a finitely many elements,
the lower bound $1/ \delta^2$ can be demonstrated by considering
qubit tomography when the density matrix does not commute with POVM elements.
We were unable to find a reference that proves this particular fact.
Ref.~\cite{nonadaptive-qubit-LB} gave an
$\Omega(\frac{1}{\delta'^2\ln(1/\delta')})$ lower bound for independent
measurements with {\em relative entropy $\delta'$} as accuracy measure
without restriction that POVM should consist of finitely many elements.

In many cases it is not necessary to determine the full state $\rho$
but only to estimate some parameters of the state.
This is an extremely general problem
which includes results such as a quantum version of the Cram\'er-Rao
bound~\cite{Helstrom69, GillM00,Hayashi-CLT} again going back to the
early prehistory of quantum information.
One special case that uses similar representation-theory techniques to our work
is the problem of spectrum estimation.
Here, the optimal covariant measurement was described by Keyl and Werner~\cite{Keyl01},
its large-deviation properties were derived in \cite{HayashiMatsumoto2002}
(see also \cite{ChristandlMitchison2006}),
and it was analyzed further in \cite{CHW07,spectrum}.
Ref.~\cite{spectrum} in particular showed (among other results)
that the Keyl-Werner algorithm required
\[
\Omega \left( \frac{d^2}{\eps^2} \right)
\leq n \leq
O\left( \frac{d^2}{\eps^2} \ln \frac d \eps \right).
\]
Our results improve the upper bound by using the same number of copies
to obtain a full estimate of $\rho$ instead of merely its spectrum.
We also improve the lower bound by showing that it applies to
{\em all} estimation strategies, not only the Keyl-Werner algorithm;
on the other hand, our lower bound is for the harder problem of state estimation,
while the lower bound of Ref.~\cite{spectrum} is for the problem of spectrum estimation.
We improve both bounds in the case when $r \ll d$.

The problem of quantum state estimation can be thought of as a special
case of minimax estimation
(i.e.~choosing an estimator that minimizes the expected loss
when we maximize over input states)
when the loss function is given by the infidelity.
Other loss functions have also been considered~\cite{Gill05,minimax}.
For example, with the 0-1 loss function (assuming $\rho$ is drawn from a finite set)
the goal is to maximize the probability of guessing $\rho$ correctly.
Here a powerful heuristic is to use the so-called
``pretty good measurement'' or PGM~\cite{Belavkin75a,Belavkin75b,PGM},
whose error is never worse than twice that of the optimal measurement for any ensemble~\cite{BK02}.
While the PGM requires a prior distribution,
prior-free versions can also be constructed~\cite{HW06}.
We will describe two closely related measurements in this paper: first,
one closely related to the PGM and then one (with roughly equivalent
performance) that corresponds precisely to a PGM over an appropriately
chosen ``uniform'' ensemble of density matrices. In each case, we
analyze the measurements directly, without making use of the results
of \cite{BK02,HW06} or other prior work.

\subsection{Sample complexity in Kueng et al.~\cite{KRT14}} \label{sec:KRT14}
The previously best achievable sample complexity for state tomography
was described in \cite{KRT14}.  Their setting does not naturally
translate into our framework, so for convenience we sketch here how that is
achievable.  First we restate one of their main theorems:
\begin{thm}
There are universal constants $C_1,C_2,C_3>0$ such that the following
holds for any $r,d$.  Let $a_1,\ldots,a_m\in\bbC^d$ be independent standard Gaussian
vectors; i.e.~normalized such that $\E[\ket{a_i}\bra{a_j}] = I_d
\delta_{ij}$.
If $m\geq C_1dr$, then with probability $\geq
1 - e^{-C_2m}$ our choice of $a_1,\ldots,a_m$ is ``good'' in a sense
we will define below.

For $X$ a matrix, define $\cA(X) = \sum_j \bra{a_j}X\ket{a_j} \ket{j}
\in \bbR^m$.   Given a $d$-dimensional density matrix $\rho$, a
vector $b\in\bbR^m$ and a noise parameter $\eta$, define $\sigma$
be any minimum of the following convex program:
\[
\min \|\sigma\|_1 \text{ subject to } \|\cA(\sigma)-b\|_2 \leq \eta.
\]
Suppose further that $\|\cA(\rho)-b\|_2 \leq \eta$.    If the
vectors $a_1,\ldots,a_m$ are good, then we have
\begin{align}
\|\rho-\sigma\|_2 \leq C_3\frac{\eta}{\sqrt{m}}.
\label{eq:KRT-estimate}
\end{align}
\end{thm}

To translate this into a quantum measurement, observe that by the
operator Chernoff bound~\cite{AW02}, we have $\frac{1}{m}\sum_{i=1}^m
\proj{a_i} \approx I_d$ with high probability. (For the purpose of
this analysis, we neglect the error here.)   We can then define a POVM
with elements $E_i = \proj{a_i} /m$.  Measuring this POVM yields
outcome $i$ with probability $p_i := \tr[E_i\rho]$; in the notation of
\cite{KRT14} we have $p = \cA(\rho)/m$.
We will define the vector $b$ of observed probabilities by measuring
$n$ independent copies of $\rho$ using this POVM.  If the resulting
vector of frequencies is $f$, i.e., outcome $i$ occurs $f_i$ times,
then we define $b = \frac{m}{n}f$.  Thus $b$ is an unbiased estimator
of $\cA(\rho)$;
i.e.~$\E[b] = \frac{m}{n}\E[f]= \frac{m}{n}np =\cA(\rho)$.
We can also estimate the error by
\[
\E \| b - \E[b] \|_2^2
= \frac{m^2}{n^2}\sum_{i=1}^m \Var[f_i]
\leq \frac{m^2}{n^2}\sum_{i=1}^m  n p_i
= \frac{m^2}{n}.
\]
We thus have $\eta \leq O(m/\sqrt{n})$ with high probability.
According to \eq{KRT-estimate} we
then have $\|\rho-\sigma\|_2 \leq O(\sqrt{m/n})=O(\sqrt{dr/n})$.
It follows that
\begin{align*}
&\|\rho-\sigma\|_1 \\
&\leq 2\sqrt{\min( \mathrm{rank}(\rho),\mathrm{rank}(\sigma) )} \| \rho - \sigma \|_2 \\
&\leq O(\sqrt{dr^2/n}).
\end{align*}
In other words, trace-distance error $\eps$ can be achieved with $n=O(dr^2/\eps^2)$.
While this bound is significantly worse than our bound of
$\tilde O(dr/\eps^2)$,
their approach does have the significant advantage of
not requiring entangled measurements.
The improved performance of our
bound (as well as that of \cite{OW-tomo})
can be seen as the advantage that entangled measurements yield for tomography.

\subsection{Two-stage measurement scheme using local asymptotic normality}
\label{sec:LAN-sample-complexity}

The local asymptotic normality in Ref.~\cite{GJK08,KahnGuta2008BookChapter}
asserts that $n$ copies $\rho_\theta^{\otimes n}$
of states $\rho_\theta$ in a sufficiently small neighborhood $\mathcal B$ of a state $\rho_{\theta=0}$
behaves like an ensemble of Gaussian states of quantum harmonic oscillators.
In relation to our discussion of state estimation,
it is important that there exists a channel~\cite{GJK08,KahnGuta2008BookChapter},
which is faithful in the limit $n \to \infty$,
from $\rho_\theta^{\otimes n}$ to gaussian states,
such that one can estimate the parameter $\theta$ of the state optimally.
The size of the neighborhood in the correspondence in fact depends on $n$.
Theorem~4.1 in Ref.~\cite{KahnGuta2008BookChapter} gives a lower bound on this size,
which reads
\[
 \mathcal B \supseteq \{ \rho ~:~ \| \rho - \rho_0 \|_2 \le n^{-1/2 + \eta} \}
\]
where $\eta \in (0, 1/6)$.

Based on this result, Ref.~\cite{KahnGuta2008BookChapter}
proposes a two-stage adaptive measurement scheme of a completely unknown state.
In the first stage, using $n_1$ copies of the state, one ``roughly''
measures the state in order to have a confidence region inside $\mathcal B$.
The actual measurement method for this first stage is not shown,
and we assume that this is a non-adaptive independent measurement on each copy.
In the second stage, one uses remaining $n_2 = n - n_1$ copies
and apply the local asymptotic normality
to optimally estimate the state.

Let us analyze the sample complexity of this proposal.
We assume that the channel between $\rho_\theta^{\otimes n}$
and Gaussian states is exactly faithful for any $n$.
This assumption may not be true on its own,
but is certainly a favorable condition to assess the advantage of local asymptotic normality.
After the first stage, the size of the confidence region must be
$\epsilon_i = n_2^{-1/2 + \eta}$ in 2-norm.
This requires at least $n_1 \ge \Omega(d^2 / \epsilon_i^2)$,
by our lower bound Theorem~\ref{thm:LBprod}.
Suppose $\epsilon_f$ is our accuracy goal in 2-norm.
If $\epsilon_f \ge \epsilon_i$, then the second stage becomes redundant,
and overall measurement is by the non-adaptive independent measurement.
Our result says that this scheme cannot be sample-optimal.
If $\epsilon_f < \epsilon_i$, then one needs $n_2 \ge \Omega( d / \epsilon_f^2 )$
in the second stage.
To achieve $\epsilon$-accuracy in 1-norm, we must have $\epsilon_f \le \epsilon / \sqrt{d}$,
and overall sample complexity becomes
\[
 n = n_1 + n_2 \ge \Omega( d^{4- 4 \eta} / \epsilon^{2-4\eta} ) + \Omega( d^2 / \epsilon^2 ) .
\]
Since $\eta < 1/6$, the dependence of $n$ on $d$
is actually worse than the independent non-adaptive scheme,
although the dependence of $n$ on $\epsilon$ is optimal.
In other words,
the measurement scheme using the asymptotic normality
may yield asymptotically optimal error scaling,
but it takes too many samples
to enter the regime where the asymptotic normality becomes useful
for high dimensional states.

%%%%%%%%%%%%%%%%%%%%%%%%%%%%%%%%%%%%%%%%%%%%%%
\section{Review on representation theory of unitary and symmetric groups}
%%%%%%%%%%%%%%%%%%%%%%%%%%%%%%%%%%%%%%%%%%%%%%

Schur-Weyl duality is a statement regarding joint representations of a
matrix group and the symmetric group.  This is standard
material~\cite{FultonHarris} in representation theory, but for the
reader's convenience we explain parts that are relevant to our
results.

Consider the Hilbert space
$\mathcal H = (\mathbb C^d)^{\otimes n}$ of $n$ qudits of
$d$-dimensions.  This space admits representations of the general
linear group $GL(d)$ and the symmetric group $\mathbb S_n$.  The
matrix group acts by simultaneous ``rotation'' as $U^{\otimes n}$ for
any $U \in GL(d)$, and the symmetric group acts by permuting tensor
factors.  Concretely, a permutation $\pi \in \mathbb S_n$ is
represented by
\[
P_\pi = \sum_{\{j_i\}} \ket{j_{\pi^{-1}(1)} j_{\pi^{-1}(2)} \cdots j_{\pi^{-1}(n)}}\bra{j_1 j_2 \cdots j_n} .
\]
Two actions $U^{\otimes n}$ and $P_\pi$ obviously commute with each other,
and hence $\mathcal H$ admits a representation of $G = GL(d) \times \mathbb S_n$.
The Schur-Weyl duality states that these two representations are commutants of each other
on $\mathcal H$.
That is, if a matrix $K$ on $\mathcal H$ commutes with all $P_\pi$,
then $K = \sum_i c_i U_i^{\otimes n}$ for some $U_i \in GL(d)$ and numbers $c_i \in \mathbb C$.
Conversely, if a matrix $K$ on $\mathcal H$ commutes with all $U^{\otimes n}$,
then $K = \sum_\pi c_\pi P_\pi$ for some $c_\pi \in \mathbb C$.

Generally, an irreducible representation (irrep) of $G$ is given by the tensor product of
an irrep of $GL(d)$ and an irrep of $\mathbb S_n$.
Since the two groups are mutual commutants on $\mathcal H$,
the irreps in $\mathcal H$ of the two groups must be in a one-to-one correspondence.
They are specified by Young diagrams,
or equivalently, partitions $\lambda = (\lambda_1,\ldots, \lambda_n)$ of $n = \sum_i \lambda_i$,
where $\lambda$ is sorted to be non-increasing.
Thus, we have a decomposition
\[
 (\mathbb C^d)^{\otimes n}
 = \bigoplus_{\lambda \vdash n} \Pi_\lambda  (\mathbb C^d)^{\otimes n}
 = \bigoplus_{\lambda \vdash n} \mathcal \cQ_\lambda \otimes \cP_\lambda
\]
where $\cQ_\lambda$ is the irrep of $\GL(d)$ and $\cP_\lambda$ is the irrep of $\mathbb S_n$,
and $\Pi_\lambda$ is the projector onto the component $\cQ_\lambda \otimes \cP_\lambda$.
Direct consequences of the decomposition are that
\begin{align}
&\Pi_\lambda X^{\otimes n} \Pi_\lambda \cong \bq_\lambda(X) \otimes \mathrm{id}_{\mathcal P_\lambda} \\
&\Pi_\lambda X^{\otimes n} = X^{\otimes n} \Pi_\lambda
\end{align}
for {\em any} $d \times d$ matrix $X$,
where we have defined $\bq_\lambda(X)$ to mean the representing matrix of $X$.
In fact, this is the main reason we are dealing with $GL(d)$,
which is dense in the set of all matrices, rather than the more familiar $\mathbb U(d)$.
The space $\cQ_\lambda$ is also an irrep of the unitary group $\mathbb U(d)$,
and our discussion of Schur-Weyl duality could have been formulated entirely with $\mathbb U(d)$;
however, under this formulation $X$ would be restricted to be unitary.

For our results it is important to understand the characters of the
irrep $\cQ_\lambda$ of $GL(d)$.  We identify a partition $\lambda$
with a {\em Young diagram} in which there are $\lambda_i$ boxes in the
$i^{\text{th}}$ row, e.g. the diagram for $\lambda=(3,2,1,1)$ is as
follows
\[
\yng(3,2,1,1).
\]
Define a Young tableau $T$ with shape
$\lambda$ to be a way of filling each box in $\lambda$ with a number,
e.g.
\[
\young(552,123,4,1).
\]
A {\em standard Young tableau} (SYT) is one in which each number from
$1,\ldots, n$ appears exactly once and numbers strictly increase from
left to right and from top to bottom, while in a semi-standard Young
tableau (SSYT) numbers weakly increase from left to right and strictly
increase from top to bottom.
Associated with a standard Young tableau $T$
there are two subgroups $A_T$ and $B_T$ of $\mathbb S_n$.
$A_T$ is the set of all permutations that permute numbers within the rows of $T$,
and $B_T$ is the set of all permutations that permute numbers within the columns of $T$.
The Young symmetrizer is then defined as
\[
Y_T = \sum_{a \in A_T, b \in B_T} \sgn(b) P_a P_b .
\]
It can be shown that $Y_T$ is proportional to an orthogonal projector,
and it turns out that $Y_T \mathcal H$ is an irrep of $GL(d)$ and is isomorphic to $\cQ_\lambda$.
Since every $T$ with the same $\lambda$ gives rise to an isomorphic irrep of $GL(d)$,
let us set $T$ to be the SYT where $1,2,\ldots,n$ are written
in order from the upper left box towards right and down.
To understand the basis of $\cQ_\lambda$,
let $\ket 1, \ket 2, \ldots, \ket d$ form the standard orthonormal basis of $\mathbb C^d$.
We may regard each basis vector $\ket E =\ket{j_1, \ldots, j_n}$ of $\mathcal H$
as a Young tableau $E$ of shape $\lambda$.
The Young symmetrizer $Y_T$ projects this basis vector to a vector of
$\cQ_\lambda$.
If there is any repetition along a column of $E$,
then $Y_T$ will annihilate it,
thanks to the antisymmetric sum over $P_b$ for $b\in B_T$.
It follows that $\cQ_\lambda = 0$ whenever $\lambda$ has more than $d$ rows.
More precisely, let $\nu_i = \nu_i(E)$ denote the number of times the
basis element $\ket i$ appears
in the tableau $E$ (also known as the {\em weight} of $E$), and let $\nu^{\downarrow}$ be the vector obtained
by sorting $\nu$ into non-increasing order.
Then $Y_T$ annihilates $E$ whenever  $\sum_{i=1}^m \nu^{\downarrow}_i > \sum_{i=1}^m \lambda_i$
for some $m = 1,\ldots, d-1$.
The negation of the last condition is often denoted as
\[
\nu \prec \lambda \Leftrightarrow
\begin{cases}
\sum_{i=1}^m \nu^\downarrow_i \le \sum_{i=1}^m \lambda_i \quad (1 \le m < d) \\
\sum_{i=1}^d \nu^\downarrow_i = \sum_{i=1}^d \lambda_i
\end{cases}
\]
and we say that {\em $\nu$ is majorized by $\lambda$}.
The surviving tableaux $E$ with $\nu(E) \prec \lambda$ form a spanning
set for $\cQ_\lambda$, or if we restrict to SSYT, they form a basis.

Now we can derive an expression for the characters of $\cQ_\lambda$.
Since $\tr \bq_\lambda(X)$ must be a function of eigenvalues of $X$,
we may assume without loss of generality that $X$ is a diagonal matrix
with eigenvalues $x_1, \ldots, x_d$
associated with the standard basis elements $\ket 1, \ldots, \ket d$.
The basis vectors of $\cQ_\lambda$ we just constructed are eigenvectors of diagonal $X^{\otimes n}$;
$X^{\otimes n} Y_T \ket E
= x_1^{\nu_1} \cdots x_d^{\nu_d} Y_T \ket E
=: x^\nu Y_T \ket E$,
where $x^\nu := x_1^{\nu_1} \cdots x_d^{\nu_d}$.
Hence, the character value $\tr \bq_\lambda(X)$ is the sum of these eigenvalues:
\begin{align}
\tr \bq_\lambda(X) = \sum_{\nu} K_{\lambda\nu} x^\nu =: s_\lambda(x).
\label{eq:Schur-polynomial}
\end{align}
Here $K_{\lambda\nu}$ is called the Kostka number
and denotes the number of SSYT with weight $\nu$ and shape $\lambda$.
One can show that $K_{\lambda\nu}>0$ if and only if $\nu\prec \lambda$.
We also define here the {\em Schur polynomial} $s_\lambda(x)$,
which is a homogeneous polynomial in $d$ variables of degree $\sum_i \nu_i = n$.
Because the character $\tr \bq_\lambda(X)$ depends only on the eigenvalues,
we will overload notation and denote this character also by
$s_\lambda(X)$.
For the same reason, it follows that $s_\lambda(XY) = s_\lambda(YX)$.
The number of terms of the Schur polynomial is equal to
\[
s_\lambda(\id_d) =
\tr \bq_\lambda( \mathrm{id}_{d} ) =
 \dim \cQ_\lambda = \prod_{i < j} \frac{\lambda_i - \lambda_j + j - i}{j-i}.
\]

%%%%%%%%%%%%%%%%%%%%%%%%%%%%%%%%%%%%%%%%%%%%
\section{Bound on Schur polynomials}
%%%%%%%%%%%%%%%%%%%%%%%%%%%%%%%%%%%%%%%%%%%%
%
\begin{lem}
Let $\rho$ and $\sigma$ be $d \times d$ density matrices.
Suppose $\rho$ has rank $r$.
Then, the character function $s_\lambda$ of
the unitary group representation labeled by Young diagram $\lambda$
satisfies
\begin{align}
s_\lambda( \rho \sigma )
&
\begin{cases}
\le  (\dim \cQ_\lambda) e^{-2nH(\bar \lambda)} F^{2n} \\
= 0 \quad \text{if} \quad \lambda_{r+1} > 0 ,
\end{cases}
\label{eq:chi-bound}
\end{align}
where
\begin{align}
 F = F( \rho, \sigma ) = \tr \sqrt{\sqrt{\rho}\ \sigma \sqrt{\rho}}
\end{align}
is the fidelity, and $H(\bar \lambda) = -\sum_i \bar \lambda_i \ln \bar \lambda_i$
is the Shannon entropy of $\bar \lambda = \lambda /n$.
\label{lem:chi-bound}
\end{lem}

\begin{proof}
Consider a positive semi-definite matrix $X$ and a number $k \ge 0$.
The largest term in the Schur polynomial $s_\lambda(X^k)$ at eigenvalues $x_1\ge \cdots \ge x_d \ge 0$ of $X$ is
\[
 x_1^{k \lambda_1} \cdots x_d^{k \lambda_d} = e^{-n k H(\bar \lambda) } e^{-n k D( \bar \lambda \| \bar x )} (\tr X)^{kn}
\]
where $\bar x = (x_1, \ldots, x_d)/\tr(X)$, and $D(p\|q)=\sum_i
p_i\ln(p_i/q_i)$ is the relative entropy.
This is because majorization implies that
\[
\max_{\nu \prec \lambda} x^\nu = x^\lambda,
\]
i.e.~the maximum is attained by putting the largest number $x_1$ with
the largest possible exponent $\nu_1 = \lambda_1$
and the second largest $x_2$ with $\nu_2 = \lambda_2$ and so on,
subject to the majorization condition $\nu \prec \lambda$.

It follows that
\begin{align}
s_\lambda( X^k ) \le \dim \cQ_\lambda \cdot e^{-n k H(\bar \lambda) } e^{-n k D( \bar \lambda \| \bar x ) } ( \tr X )^{kn} .
\end{align}
Now, we set $X = \sqrt{\sqrt{\rho}\  \sigma \sqrt{\rho}}$ and observe
$s_\lambda( \rho \sigma ) = s_\lambda( X^2 )$.  Using the fact that
$D(\bar \lambda \| \bar x)$ is always non-negative and
$=+\infty$ when the rank of $\bar \lambda$ is larger than that of
$\bar x$, we arrive at Eq.~\eqref{eq:chi-bound}
\end{proof}

Note that since $s_\lambda( \bar \lambda )$ is a sum of non-negative terms,
it is lower bounded by its largest term:
\begin{align}
s_\lambda ( \bar \lambda ) \ge e^{-n H(\bar \lambda)} .
\label{eq:chi-lower-bound}
\end{align}

\section{Tomography}

Suppose we are given with $\rho^{\otimes n}$, $n$ copies of an unknown density matrix $\rho$.
What is the best strategy to learn about $\rho$?
The input state has a trivial symmetry $\mathbb S_n$ under the permutations of the tensor factors.
So, the POVM elements of the optimal strategy can be taken to commute
with $P_\pi$ without loss of generality.  Additionally since we do not
assume any distribution over $\rho$, our measurement should not
perform differently when $\rho$ is replaced by $U\rho U^\dag$.  This
means that if $M_\sigma$ is the outcome corresponding to $\sigma$ then
we should have
\[
M_{U\sigma U^\dag} = (U^\dag)^{\ot n} M_\sigma U^{\ot n}.
\]
These observations, along with the Schur-Weyl decomposition,
motivate us to define positive semi-definite operators
\begin{equation}
  \label{eq:POVM1}
  M(\lambda,U) := \frac{\dim \cQ_\lambda}{s_\lambda( \bar \lambda )}
  \Pi_\lambda (U \bar \lambda U^\dagger)^{\otimes n} \Pi_\lambda ,
\end{equation}
for each unitary $U$ and Young diagram $\lambda$ that partitions $n$ with at most $d$ rows.
As before, $\bar \lambda$ denotes the diagonal matrix with entries
$\lambda /n$.

We first show that the $M(\lambda,U)\rd U$ constitute a POVM,
where $\rd U$ is the Haar probability measure on $\mathbb U(d)$.
It suffices to check $\int \rd U M(\lambda,U) = \Pi_\lambda$,
for $\sum_\lambda \Pi_\lambda = I$.
Since $\int \rd U M(\lambda, U)$ is invariant under any unitary conjugation or permutation,
we only need to check the traces of both sides.
\begin{align*}
\int \rd U \tr M(\lambda,U)
&= \frac{\dim \cQ_\lambda \dim \cP_\lambda}{s_\lambda( \bar \lambda )} \int \rd U  \tr \bq_\lambda(U \bar \lambda U^\dagger) \\
&= \frac{\dim \cQ_\lambda \dim \cP_\lambda}{s_\lambda( \bar \lambda )} \int \rd U \tr \bq_\lambda(\bar \lambda) \\
&= \tr \Pi_\lambda
\end{align*}

\begin{figure}
\includegraphics[width=.45\textwidth]{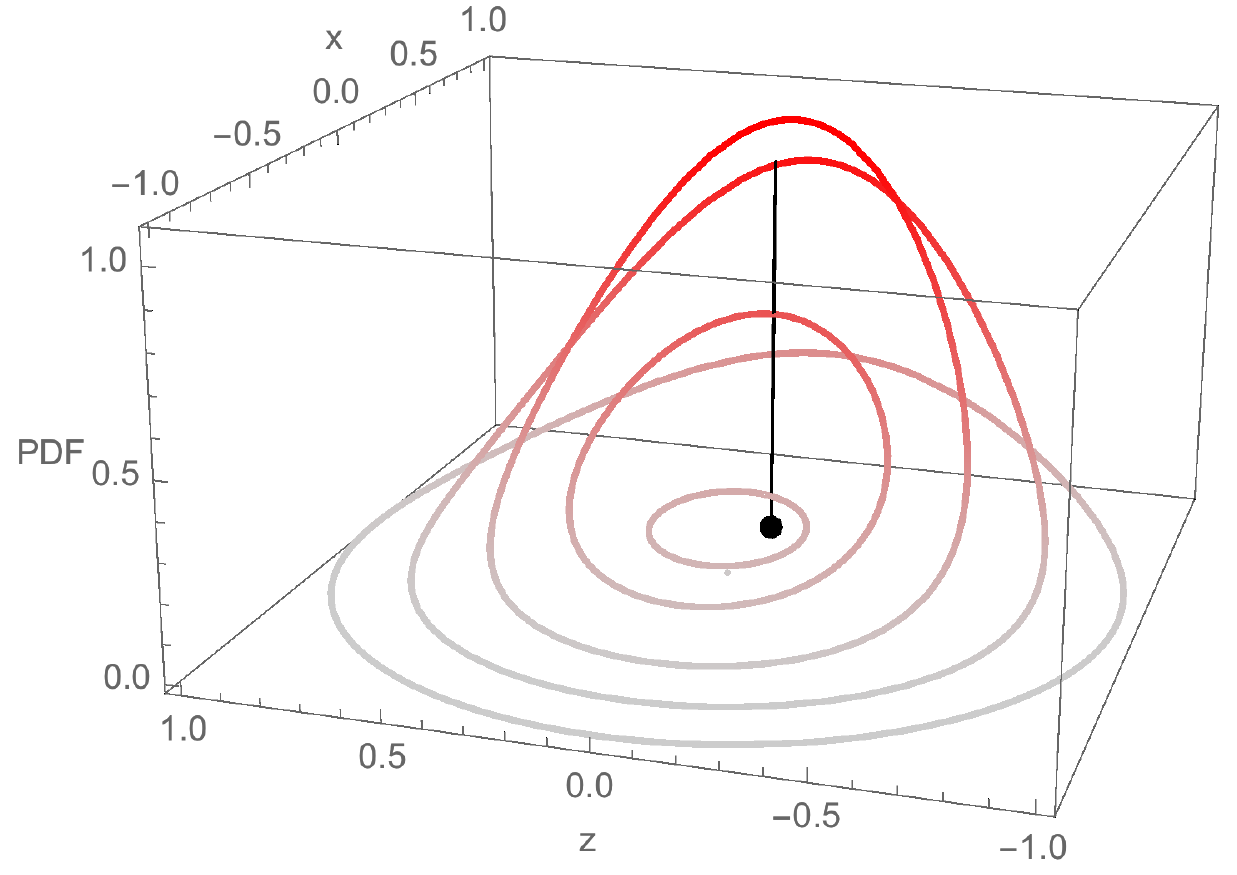}
(a)~$n=10$
\includegraphics[width=.45\textwidth]{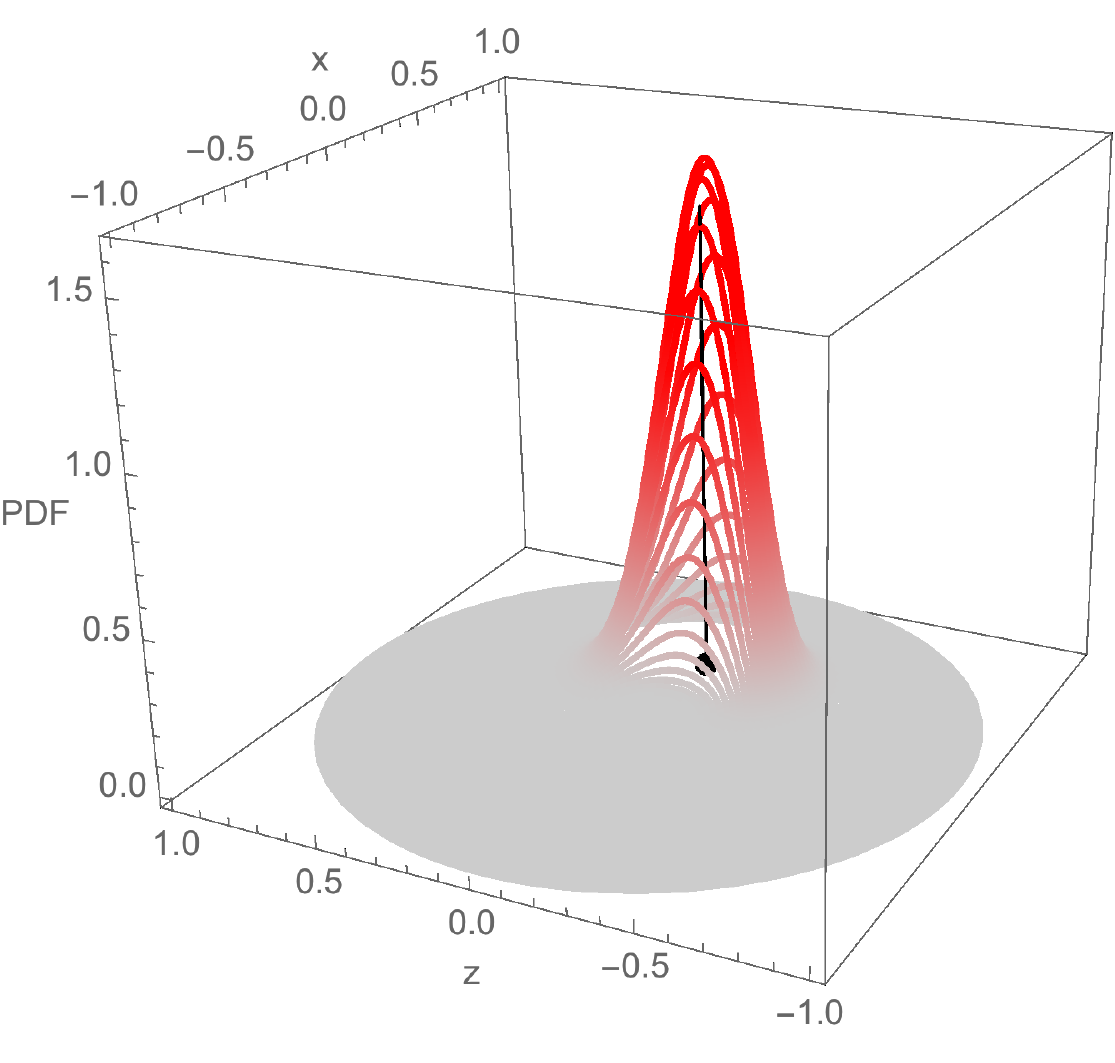}
(b)~$n=100$
\caption{Measurement outcome probability density functions (PDF) of
the POVM in Eq.~\eqref{eq:POVM1} on $n$ copies
of a qubit state $\rho = 0.7 \proj{0} + 0.3 \proj{1}$,
represented by the black delta function.
The PDF is plotted over the Bloch states
$\hat \rho = \frac12(I+z \sigma_z + x \sigma_x)$ with $x^2 + z^2 \le 1$,
and is zero except on the circles
because output states are of form $U \bar \lambda U^\dagger$
where $\bar \lambda$ is a density matrix from a discrete set,
which becomes finer as $n$ increases.
Red is the confidence region, which becomes small for large $n$.
%For a general $d$-dimensional state,
%the confidence region of probability (confidence level) $99\%$
%has size less than $\delta$, measured in ``infidelity'' $1-F$,
%if $n \simeq (3d^2/\delta) \ln(15d^2 / \delta ) + 3/\delta$.
}
\label{fig:UB}
\end{figure}

Next, we bound the probability density of measuring $M(\lambda, U)$.
Let $F = F(\rho, U \bar \lambda U^\dagger)$ be the fidelity.
We claim
\begin{align}
\tr( M(\lambda,U)\rho^{\otimes n}) \le (n+1)^{2dr} F^{2n},
\label{eq:meas-prob}
\end{align}
where $r$ is the rank of $\rho$.

To show this, we need a bound on $\dim \cP_\lambda$:
\begin{align}
 \dim \cP_\lambda \le e^{n H(\bar \lambda) } , \label{eq:Sn-dim-bound}
\end{align}
which has implicitly appeared in \cite{ChristandlMitchison2006}.
This follows from
\begin{align}
 \dim \cP_\lambda \prod_i \bar \lambda_i^{\lambda_i}
 \le \frac{n!}{\prod_i \lambda_i!} \prod_i \bar \lambda_i^{\lambda_i}
 = \frac{n!}{n^n} \frac{\prod_i \lambda_i^{\lambda_i}}{\prod_i \lambda_i!}
 \le 1.
 \label{eq:Sn-dim-bound-pf}
\end{align}
The first inequality is by the ``hook length formula''~\cite{FultonHarris}.
For the last inequality we note that the function $f(z) = z \ln z - \ln \Gamma(z+1)$
satisfies $f(0)=0$ and $f ''(z) > 0$ for $z  > 0$~\cite{Batir2008Inequalities}.
Hence, $\sum_{i=1}^d f(\lambda_i)$ with $\sum_{i=1}^d \lambda_i = n$ is maximum
if and only if $\lambda_1 = n$, in which case the inequality is saturated.

Eqs.~\eqref{eq:chi-lower-bound} and \eqref{eq:Sn-dim-bound} now imply that
\begin{align*}
\tr(M(\lambda,U) \rho^{\otimes n} )
&=  \frac{\dim \cQ_\lambda \dim \cP_\lambda}{s_\lambda( \bar \lambda )}
  s_\lambda( \rho U \bar \lambda U^\dagger ) \\
&\le \dim \cQ_\lambda \cdot e^{2n H(\bar \lambda)} s_\lambda( \rho U \bar \lambda U^\dagger ).
\end{align*}
By Eq.~\eqref{eq:chi-bound}, this is nonzero only if $\lambda_{r+1}=\lambda_{r+2}= \cdots =\lambda_d = 0$.
In this case, we have $\dim \cQ_\lambda \le (n+1)^{dr}$, and arrive at
Eq.~\eqref{eq:meas-prob}.

The output of our POVM is $\hat \rho = U \bar \lambda U^\dagger$.
The probability of obtaining $\hat \rho$ where $\hat \rho$ has small fidelity,
say infidelity $\delta$, to the true state $\rho$
can be estimated by integrating Eq.~\eqref{eq:meas-prob} over all pairs $(\lambda, U)$
such that $F(\rho, U \bar \lambda U^\dagger) \le 1 - \delta$.
Since $\sum_\lambda \int  \rd U < (n+1)^d$, we see that
\begin{align}
\Pr[ ~F(\hat \rho, \rho) \leq 1 - \delta~ ] \leq  (n+1)^{3dr} e^{-2n \delta} .
\end{align}

\subsection{Pretty Good Measurement}

Here we propose another POVM that achieves the same (up to constants)
sample-complexity for tomography.

Recall that given an ensemble
$\{(p_1,\phi_1), \ldots, (p_m, \phi_m)\}$, the PGM has measurement
operators $M_i := \bar\phi^{-1/2} p_i \phi_i \bar \phi^{-1/2}$
with $\bar\phi := \sum_i p_i \phi_i$~\cite{PGM}.
A relevant ensemble for us is the one in which $\phi_i$ is equal to $\sigma_i^{\otimes n}$,
and the index $i$ should run over all state space;
our ensemble is determined by $n$ and a probability measure $\rd \sigma$ on the whole state space $\{ \sigma \}$.
Demanding the unitary invariance of $\rd \sigma$,
we have
\begin{align}
 \bar \phi &= \int \rd \sigma ~\sigma^{\otimes n}
 = \sum_\lambda \frac{ \int \rd \sigma s_\lambda( \sigma ) }{ \dim \cQ_\lambda } \Pi_\lambda , \nonumber\\
M_\sigma \rd \sigma
 &= \sum_\lambda \frac{\dim \cQ_\lambda}{\E s_\lambda}
     \Pi_\lambda \sigma^{\otimes n} \Pi_\lambda \rd \sigma  ,\label{eq:POVM2}
 \end{align}
where $\E s_\lambda = \int \rd \sigma s_\lambda(\sigma)$.
It follows that the probability density of measuring $M_\sigma$ given a state $\rho$ of rank at most $r$ is
\begin{align*}
 \tr(M_\sigma \rho^{\otimes n} ) \rd \sigma
 &= \sum_\lambda \frac{(\dim \cQ_\lambda \cdot \dim \cP_\lambda )s_\lambda( \sigma \rho )}{\E s_\lambda} \rd \sigma \\
 &\leq
 \sum_{\lambda : \lambda_{r+1}=0}
 \frac{(\dim \cQ_\lambda)^2}{e^{n H(\bar \lambda)} \E s_\lambda } F^{2n} \rd \sigma
\end{align*}
where the inequality is by Eq.~\eqref{eq:chi-bound} and \eqref{eq:Sn-dim-bound}.
This is the same scaling in $n$ up to constants as Eq.~\eqref{eq:meas-prob},
provided
\[
e^{n H(\bar \lambda)} \E s_\lambda \geq (nd)^{-O(dr)} .
\]
Indeed we show that this is the case
if we choose a uniform distribution over the simplex of spectra of $\sigma$.
First, we bound the Schur polynomial by its largest term:
\begin{align*}
\int \rd \sigma s_\lambda( \sigma )
 &\ge
 \frac{1}{f_d(\vec \lambda = 0)}\underbrace{
 \int_{s_i\ge 0,~ \sum_i s_i = 1}  s_1^{\lambda_1} \cdots s_d^{\lambda_d}  \rd s}_{f_d(\vec \lambda)} .
\end{align*}
By writing the integral explicitly, we see that
\begin{align*}
f_d( \lambda_1,\ldots,\lambda_d )
&= f_2\left(\lambda_1, d-2+\sum_{i=2}^d \lambda_i\right) f_{d-1}(\lambda_2, \ldots, \lambda_d),\\
f_2( a, b) &= \frac{a! b!}{(a+b+1)!} \quad \quad \text{(Eq.~\eqref{eq:Dirichlet} below)}.
\end{align*}
This implies that $f_d(\vec \lambda) = \lambda_1 ! \cdots \lambda_d ! / ( n+d-1 )!$.
(We just calculated the normalization factor for the Dirichlet distribution.)
Hence,
\begin{align*}
e^{n H(\bar \lambda)} \int \rd \sigma ~s_\lambda( \sigma )
&\geq
e^{n H(\bar \lambda)}  \frac{\lambda_1! \cdots \lambda_d ! (d-1)!}{(n+d-1)!}\\
&\geq
(n+d)^{-d} ,
\end{align*}
where in the second inequality we use Eq.~\eqref{eq:Sn-dim-bound-pf}.
We conclude that this PGM defined by the uniform spectrum distribution
achieves the same bound (up to constants) on the sufficient number of copies for tomography.

Both POVMs in Eqs.~\eqref{eq:POVM1} and~\eqref{eq:POVM2}
are inspired by the pretty good measurement, and
indeed the measurement operator corresponding to the estimate $\sigma$ is like a
distorted version of $\sigma^{\ot n}$.
Variants of the PGM have
been proposed in which the measurement operators are distorted
versions of higher powers of the state $p_i\sigma_i$, i.e.~$M_i =
X^{-1/2} (p_i\sigma_i)^k X^{-1/2}$ where $X \equiv \sum_i (p_i
\sigma_i)^k$.  When $k=1$ this is the PGM, but the cases $k=2$ and
$k=3$ have also been found useful in specific settings; see
\cite{Tyson09} for a review. If we take $k\ra \infty$ here then
this corresponds precisely to the Keyl ``rotated-highest-weight''
strategy~\cite{Keyl06}.  It is possible that this framework could be used to
formally compare the performance of these different strategies.

\begin{proof}[A definite integral]
Here we show for $a,b>0$
\begin{align}
\int_0^1 u^{a-1} (1-u)^{b-1} \rd u = \frac{\Gamma(a)\Gamma(b)}{\Gamma(a+b)}.
\label{eq:Dirichlet}
\end{align}
Let $x = ut \ge 0$ and $y = (1-u)t \ge 0$.
The Jacobian is $|\partial(x,y)/\partial(u,t)| = t$.
Then,
\begin{align*}
\Gamma(a) \Gamma(b)
&= \int_0^\infty \int_0^\infty \rd x \rd y ~x^{a-1} y^{b-1} e^{-x-y}\\
&= \int_0^1 \rd u \int_0^\infty \rd t ~t^{a+b-1} u^{a-1}(1-u)^{b-1} e^{-t}\\
&= \Gamma(a+b) \int_0^1 \rd u ~u^{a-1} (1-u)^{b-1}.
\end{align*}
\end{proof}

\section{Lower bounds} \label{sec:net}

\begin{thm}\label{thm:LB}
Let $\eps \in (0,1)$ and $\eta \in (0,1)$.
Suppose there exists a POVM $\{M_{\sigma} \rd\sigma\}$ on
$(\bbC^d)^{\otimes n}$ such that for
any state $\rho \in \mathbb C^{d \times d}$ with rank $\leq r$,
\begin{align}
 \int_{\frac 1 2 \|\sigma - \rho\|_1 \leq \eps/2} \rd \sigma
   \tr[M_\sigma\rho^{\otimes  n}]
   \geq 1-\eta.
\label{eq:meas-too-good}
\end{align}
Then,
\[
n\geq C \frac{dr}{\eps^2} \frac{(1- \eps)^2}{\ln( d /r\epsilon)}
\] for $C$ a constant depending only on $\eta$.
In addition, if $r = d$, then
\[
n \geq C \frac{d^2}{\eps^2} (1-\eps)^2
\] for $C$ a constant depending only on $\eta$.
\end{thm}
This theorem implies that achieving infidelity $\delta = 1- F$ requires
$n\geq \tilde \Omega(dr/\delta)$. For both trace distance and fidelity these
lower bounds match our upper bounds up to the log factors.

Let us say that a POVM $M_\sigma$ on $(\bbC^d)^{\otimes n}$ is
an independent measurement if it is equal to the tensor product of $n$ POVM's
$M^{(a)}$ on $\bbC^d$. Then, we have
\begin{thm}\label{thm:LBprod}
Let $\delta \in (0,1)$ and $\eta \in (0,1)$.
Suppose there exists an independent measurement $M_\sigma \rd \sigma$ on
$(\bbC^d)^{\otimes n}$ such that for
any state $\rho \in \mathbb C^{d \times d}$ with rank $\leq r$,
\begin{align}
 \int_{1-F(\sigma,\rho) \leq \delta / 4} \rd \sigma
   \tr[M_\sigma \rho^{\otimes  n}]
   \geq 1-\eta.
\label{eq:meas-too-good-prod}
\end{align}
Then,
\[
n \geq C \frac{dr^2}{\delta^2 \ln (2/\delta)} (1- \delta)^4
\] for $C$ a constant depending only on $\eta$.
In addition, given $\epsilon \in (0,1)$,
if the independent measurement $M_\sigma \rd \sigma$ satisfies
\begin{align}
 \int_{\frac12 \| \rho - \sigma \|_1 \leq \epsilon / 2} \rd \sigma
   \tr[M_\sigma \rho^{\otimes  n}]
   \geq 1-\eta
\label{eq:meas-too-good-prod2}
\end{align}
for any state $\rho \in \mathbb C^{d \times d}$ of possibly full rank, then
\[
n \ge C \frac{d^3}{\epsilon^2} (1 - \epsilon)^2
\]  for $C$ a constant depending only on $\eta$.
\end{thm}
Note that the fidelity lower bound implies trace distance bound
\[
n \geq C \frac{dr^2}{\epsilon^2 \ln (2/\epsilon)}.
\]

\begin{proof}
We will show that any measurement satisfying \eqref{eq:meas-too-good},
\eqref{eq:meas-too-good-prod}, or \eqref{eq:meas-too-good-prod2}
will imply the existence of a communication protocol that can reliably
send a large message.  Holevo's theorem~\cite{Holevo73} can then be
used to obtain a lower bound on $n$.
The independent measurement case is very similar and
will be explained at the end of this proof.

Following convention, call the sender Alice and the receiver Bob.
We will show in \lemref{eps-net} below that there exists a states
$\rho_1,\ldots,\rho_N$ each with rank $\leq r$ such that
\begin{align}
 \frac 1 2 \|\rho_i -\rho_j\|_1 > \eps\quad \forall i\neq j.
\label{eq:packing}
\end{align}
The set $\{\rho_1,\ldots,\rho_N\}$ is known as an $\eps$-packing net.
Fix such a net, along with
a measurement $\{M_\sigma\rd\sigma\}$ satisfying \eq{meas-too-good}.

We will now construct a communication protocol.  Alice will choose a
message $x\in [N] := \{1,\ldots,N\}$ which she will encode by sending
$\rho_x^{\otimes n}$.  Bob will use the state estimation scheme
$\{M_\sigma\}$ to attempt to guess $x$.  If $\sigma$ is within $\eps/2$
trace distance of some $\rho_{y}$ then Bob will guess $y$.
By \eq{packing}, there is always at most one $\rho_{y}$
satisfying this condition.  If no such $\rho_{y}$ exists, Bob
will output failure.   This results in the POVM with measurement
outcomes
\begin{align}
\tilde M_{y} &= \int_{\frac 1 2 \|\sigma-\rho_{y}\|_1 \leq  \eps/2} \rd \sigma M_\sigma \\
 \tilde M_{\text{fail}} &= \id - \sum_{y \in [N]} \tilde M_{y}.
\end{align}
Define $\Pr[y | x] = \tr[\tilde M_{y}\rho_x^{\ot n}]$.
From \eq{meas-too-good} we have that $\Pr[x|x] \geq 1-\eta$.  In other
words, Bob has a $\geq 1-\eta$ chance of correctly decoding Alice's
message.  By Fano's inequality~\cite{Fano61}, this implies that
\begin{align}
 I(X:Y) \geq (1-\eta)\ln(N) - \ln(2).
\end{align}

On the other hand, Holevo theorem~\cite{Holevo73} states that
$I(X:Y)\leq \chi$ where $\chi$ is the Holevo information:
\begin{align}
 \chi = S\left( \frac1N \sum_{x\in [N]} \rho_x^{\otimes n} \right) -
  \frac1N \sum_{x\in [N]} S(\rho_x^{\otimes n}) .
\label{eq:chi}
\end{align}
In \lemref{eps-net} below we will argue that there exists a packing
net with large $N$ and small $\chi$. Specifically, we will bound $\chi \leq n \chi_0$
where
\[
\chi_0 = S\left( \mathbb E_U U \rho_x U^\dagger \right) - S(\rho_x) ,
\]
for an appropriate Haar random unitary $U$, and prove $\chi_0 = \tilde O( \epsilon^2)$.
This will imply that
\[ n \geq \frac{(1-\eta)\ln(N) - \ln(2)}{\chi_0}. \]
Our result then follows from \lemref{eps-net} below.

\begin{figure}
\includegraphics[width=.47\textwidth]{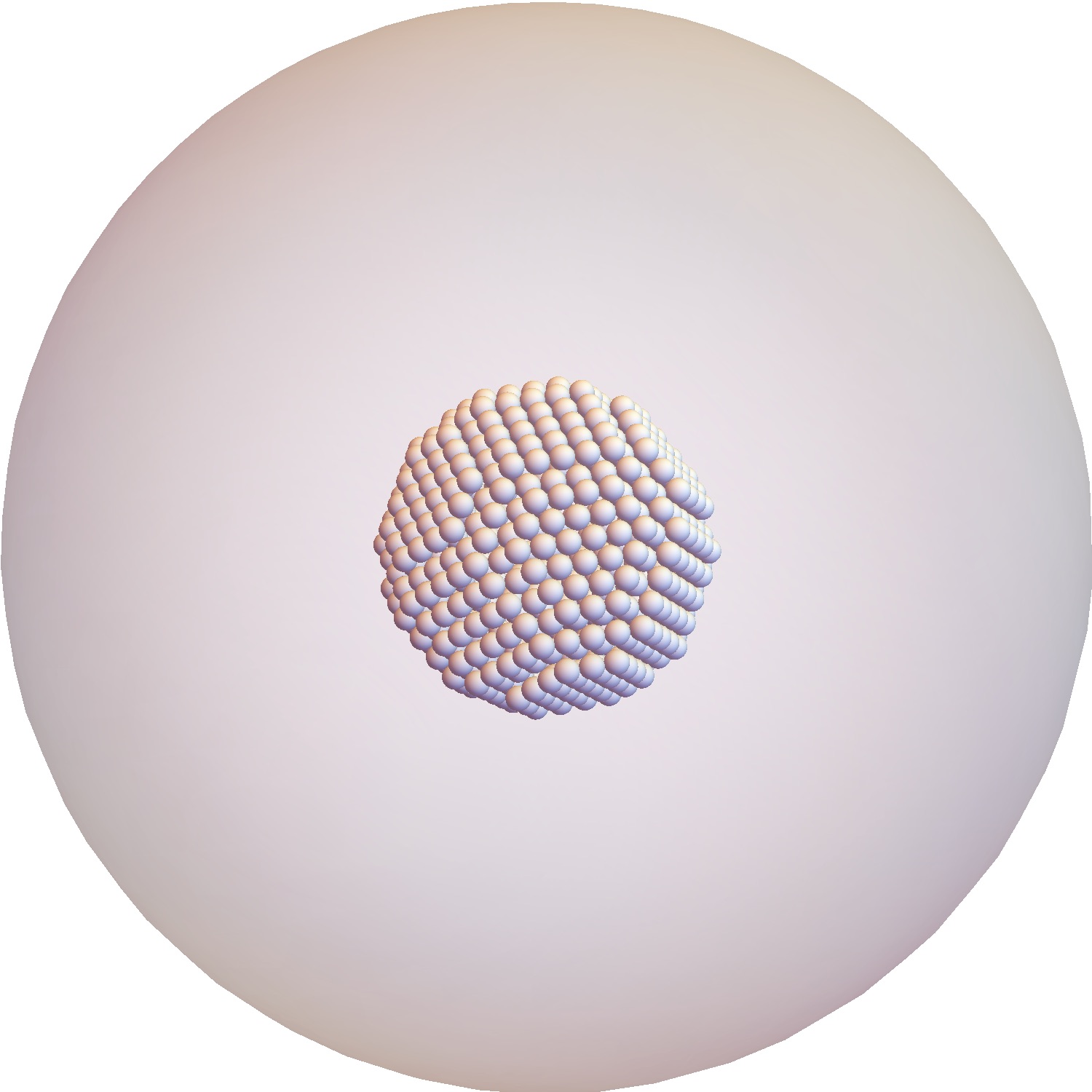}
\caption{Packing net in a small region of state space.
The centers of the small balls represent states $\rho_x$ that
are separated from each other by distance $a$,
which is larger than the resolution of tomography.
This enables a communication channel using $\rho_x$.
The packing net has diameter $10a \ll 1$, and contains $N=10^D$ balls,
where $D=d^2-1$ is the dimension of the state space,
making the channel capacity of order $D$.
Meanwhile, Holevo information for the packing net is proportional to $a^2$.
This establishes our lower bound on the general tomography.
For visualization purpose, we depict $D=3$ case
with the packing net centered around the maximally mixed qubit state.
}
\label{fig:packingball}
\end{figure}

For the independent measurements,
Bob has to infer the state based on the measurement outcome distribution
from each copy. Hence, the Holevo information must be calculated with respect
to the outcome distribution.
Since the construction of the states, and the calculation of Holevo information
are somewhat similar to those for the joint measurements,
we present a complete proof in Sec.~\ref{sec:indLB}
after the proof of Lemma~\ref{lem:eps-net}.
\label{incompletePf}
\end{proof}

\begin{lem}\label{lem:eps-net}
There exist $\eps$-packing nets I,II,III of $d$-dimensional states (i.e.~satisfying \eq{packing})
characterized in the following table.
\begin{center}
\begin{tabular}{|c|c|c|c|c|}
\hline
    & {\rm rank} & $ \chi_0/c \le $ & $c \ln N \ge$ & {\rm restriction}\\
\hline
I & $r$ & $\epsilon^2 \ln (d/r \epsilon)$ & $rd$ & $\epsilon \le 2^{-4}$, $ r < d/3$ \\
\hline
II & $d$ & $\epsilon^2$ & $d^2$ & $\epsilon \le 2^{-3}$, $d$ even \\
\hline
III & $r$ & $\ln(d/r)$ & $rd (1 - \epsilon )$& $r < d(1-\epsilon)/6$\\
\hline
\end{tabular}
\end{center}
where $c>0$ is a sufficiently large constant; $c=1000$ is good enough.
\end{lem}

We remark that packing nets of size $\exp(\Omega(dr))$ for rank-$r$ states
have been achieved as early as 1981~\cite{Szarek1981Nets, Szarek1983finite};
see also \cite{Winter04,LVWW} which used
them for applications in communication complexity.  These imply an
$\Omega(dr)$ lower bound on the number of copies needed when $\eps$ is
constant~\cite{Winter04,LVWW,OW-tomo} and has been used in
\cite{FlammiaGrossLiuEtAl2012}
to argue an $\tilde\Omega(r^2d^2)$ lower bound on the number of copies needed
for constant accuracy using adaptive Pauli
measurements.
Our main new
contribution here is to analyze at the same time the Holevo capacity
corresponding to these ensembles, in order to obtain bounds
with simultaneously optimal scaling with $r$, $d$ and $\eps$.

\subsection{Probabilistic existence argument}

We will define a set of states $\rho_U = U \rho_I U^\dagger$
where $U$ is any element of some subgroup $G \subseteq \mathbb U(d)$.
Suppose
\[
\Pr_U [~ \| \rho_U - \rho_I \|_1 \le \epsilon ~] \le \zeta
\]
for Haar random $U \in G$.
We wish to find a set $\{ U_i \}$ of unitaries with cardinality at least $\lceil 1/\zeta \rceil$
such that $\| \rho_{U_i} - \rho_{U_j} \|_1 > \epsilon$ whenever $i \neq j$.
This can be done inductively starting with the singleton $\{ I \}$.
Since Haar measure is left-invariant, $\Pr_U [ ~\| \rho_U - \rho_{V} \|_1 \le \epsilon ~] \le \zeta$
for any unitary $V \in G$.
If $m < \lceil 1/\zeta \rceil$ unitaries are chosen, the probability of choosing a unitary $U$ such that
$\rho_{U}$ is $\epsilon$-close to any previously chosen $\rho_{U_i}$
is at most $\zeta m$, which is strictly smaller than $1$.
This proves the existence of one more desired unitary,
and we obtain a set of $\lceil 1/\zeta \rceil$ elements.
The probability $\zeta$ will be repeatedly estimated using the following fact.
\begin{lem}[Lemma III.5 of Ref.~\cite{HaydenLeungWinter2006}]
Let $P$ and $Q$ be projectors on $\mathbb C^{d}$ of rank $p$ and $q$, respectively.
Let $U \in \mathbb U(d)$ be Haar random. It holds that
\begin{align*}
\forall z > 0: \Pr_U \left[ \frac{d}{pq} \tr QUPU^\dagger  \ge 1 + z \right] &\le \exp[ - pq f(z) ],\\
\forall z \in (0,1):\Pr_U \left[ \frac{d}{pq} \tr QUPU^\dagger  \le 1 - z \right] &\le \exp[ - pq f(-z) ],
\end{align*}
where
\[
f(z) = z - \ln(1+z) \ge
\begin{cases}
 (1+z) /2 & z \in [5, \infty) \\
(1-\ln 2)\ z^2 & z \in (-1,1] \\
 z^2 /2 & z \in (-1,0]
\end{cases}.
\]
\label{lem:projector-overlap}
\end{lem}
Ref.~\cite{HaydenLeungWinter2006} does not explicitly cover the
$z > 1$ case for the first inequality,
though it implicitly covered in their proof. See Appendix~\ref{sec:randomProjectorOverlap}.

\subsection{Joint Measurement}
This section constitutes the proof of Lemma~\ref{lem:eps-net}.

\subsubsection{Packing net I}
Suppose $3r < d$. Let
\begin{align}
U = \begin{pmatrix}
I_r & 0 & 0 \\
0 & A_{r \times r} & B_{r \times (d-2r)}\\
0 & C_{(d-2r) \times r} & D_{(d-2r) \times (d-2r)}
\end{pmatrix}
\label{eq:U-components}
\end{align}
be a unitary matrix of $\mathbb U(d-r)$ with blocks as indicated,
embedded into $\mathbb U(d)$.
For $0 \le t \le 1$, define
\begin{align}
\rho_{t,I} &=
\begin{pmatrix}
(1- t^2) I_r /r& t \sqrt{1-t^2} I_r /r& 0 \\
t \sqrt{1-t^2} I_r/r & t^2 I_r /r& 0 \\
0 & 0 & 0_{d-2r}
\end{pmatrix}, \\
\rho_{t,U} &=
U \rho_{t,I} U^\dagger . \nonumber
\end{align}
It is a maximally mixed state on an $r$-dimensional subspace.
%%%%%%%%%%%%%%%%%%%%%%%%%%%%%%%%%%%%
We claim that the distance between $\rho_{t,U}$ satisfies
\begin{align}
\| \rho_{t,U} - \rho_{t, I_{d-r} } \|_1 \ge \frac{t\sqrt{1-t^2}}{r} \tr C^\dagger C
\label{eq:dist-U}
\end{align}
where $C$ is as in Eq.~\eqref{eq:U-components}.
To prove this, observe that
$\| \rho_{t,U} - \rho_{t, I_{d-r} } \|_1 \ge |\tr[ (\rho_{t,U} - \rho_{t, I_{d-r} }) V]|$
where
\[
V = \begin{pmatrix}
A & 0 & BF \\
0 & E & 0 \\
C & 0 & DF
\end{pmatrix}
\]
and $E \in \mathbb U(r)$ and $F \in \mathbb U(d-2r)$ are arbitrary.
Abbreviate as $\alpha = (1-t^2)/r$, $\beta = t\sqrt{1-t^2}/r$, and $\gamma=t^2 /r$.
Expanding the formula,
\begin{align*}
&\tr[ (\rho_{t,U} - \rho_{t, I_{d-r} }) V] \\
&=
\tr\left[ \begin{pmatrix}
0 & \beta(A^\dagger - I) & \beta C^\dagger \\
\beta(A-I) & \gamma(AA^\dagger - I) & \gamma AC^\dagger \\
\beta C & \gamma CA^\dagger & \gamma CC^\dagger
\end{pmatrix}
\begin{pmatrix}
A & 0 & BF \\
0 & E & 0 \\
C & 0 & DF
\end{pmatrix} \right]\\
&=
\tr
\begin{pmatrix}
\beta C^\dagger C & \star & \star \\
\star & (\gamma AA^\dagger -I)E & \star \\
\star & \star & (\beta CB+\gamma CC^\dagger D)F
\end{pmatrix}.
\end{align*}
For some unitary $E$ and $F$, the trace of the last two entries
become the trace norm of the matrices in the parentheses,
which are non-negative.
This proves Eq.~\eqref{eq:dist-U}.
%%%%%%%%%%%%%%%%%%%%%%%%%%%%%%%%%%%%
\begin{lem}
If $0 < t < 1/2$ and $r < d/3$,
there exists a finite subset $\{ U_i \} \subset \mathbb U(d-r)$ of cardinality $N \ge \exp(dr/54)$
such that $\| \rho_{t,U_i} - \rho_{t,U_j} \|_1 > t/4$ for any $i \neq j$.
The Holevo $\chi_0$ of $\{ \rho_{t,U_i} \}_{i=1}^N$ fulfills
$
\chi_0 \le t^2  \ln \frac{ed}{t^2 r} .
$
\end{lem}
%%%%%%%%%%%%%%%%%%%%%%%%%%%%%%%%%%%%
\begin{proof}
Lemma~\ref{lem:projector-overlap} states that
if $U$ is a Haar random unitary matrix of dimension $k$,
then any $k_1 \times k_2$ subblock $K$ of $U$ satisfies
\[
\Pr\left[ \frac{k}{k_1 k_2} \tr( K^\dagger K)  < 1-z \right]
\le \exp ( -k_1 k_2 z^2 / 2 )
\]
for $z \in (0,1)$.
Eq.~\eqref{eq:dist-U} says that $\| \rho_{t,I_{d-r}} - \rho_{t,U} \|_1 \le t / 4 $ implies
$\frac{d-r}{r(d-2r)} \tr C^\dagger C  \le \frac{1}{\sqrt{3}} < 1 - \frac{1}{3}$.
Therefore,
\[
\Pr[ \| \rho_{t,I_{d-r}} - \rho_{t,U} \|_1 \le t/4 ] \le e^{-r(d-2r) / 18 } < e^{-rd/54} ,
\]
and we resort to the probabilistic existence argument.

Next, we estimate the Holevo information $\chi$.
Since $U$ is unitary, we have $S(\rho_{t,U}) = S(\rho_{t,I_{d-r}}) = \ln r$.
By the concavity of entropy, the ensemble average may be replaced with
$\bar \rho_t = \int \rd U \rho_{t,U}$,
only to increase the entropy.
By Schur's lemma, the matrix $\bar \rho_t$ is diagonal, and has entropy
\begin{align*}
S(\bar \rho_t) = H(t^2) + (1-t^2) \ln r + t^2 \ln(d-r),
\end{align*}
where $H(t^2) = -t^2 \ln( t^2 ) - (1-t^2) \ln(1-t^2)$ is the binary entropy.
Combining, we have $\chi / n \le H(t^2) + t^2 \ln \frac{d-r}{r} $.
Using $H(z) \le z \ln (e/z)$, we finish the proof.
\end{proof}

\subsubsection{Packing nets II \& III}
Assume that $d$ is an even number,
and fix a projector $Q = \mathrm{diag}(1,\ldots,1,0,\ldots,0)$ of rank $r \le d/2$.
For any $d \times d$ unitary $U$ and $0\le t \le 1$, define
\begin{align}
\tau_{t,U} = \frac{1+t}{2r}UQU^\dagger + \frac{1-t}{2(d-r)}(I_d-UQU^\dagger) .
\label{eq:tau}
\end{align}

Given an ensemble $\{ \tau_{t,U} \}$,
the entropy of the ensemble average is certainly at most $\ln d$.
The entropy of $\tau_{t,U}$ is equal to $H( (1+t)/2 ) + \frac{1+t}{2}\ln r + \frac{1-t}{2} \ln(d-r)$,
where $H( \cdot )$ is the binary entropy.
Therefore, the Holeve $\chi_0$ is bounded as
\begin{align}
 \chi_0 \le \frac12 \ln \frac{d^2}{r(d-r)} + \frac{t}{2} \ln \frac{d-r}{r} - H \left( \frac{1+t}{2}\right) .
 \label{eq:chi-tau}
\end{align}
Next, if $A$ denotes the upper-left $r \times r$
and $C$ the lower-left $(d-r) \times r$ submatrix of $U$,
we have
\begin{align}
\tr A A^\dagger + \tr C C^\dagger &= r \nonumber \\
\tr B B^\dagger + \tr D D^\dagger &= d-r \label{eq:trCompU}\\
\tr C C^\dagger + \tr D D^\dagger &= d-r \nonumber
\end{align}
and
\begin{align*}
 &\tau_{t,U} - \tau_{t,I_d} =\\
&\begin{pmatrix}
 \alpha AA^\dagger + \beta BB^\dagger - \alpha I_r & \star \\
 \star & \alpha CC^\dagger + \beta DD^\dagger - \beta I_{d-r}
\end{pmatrix}
\end{align*}
where $\alpha = (1+t)/2r$ and $\beta = (1-t)/2(d-r)$.
Multiplying a unitary $\mathrm{diag}( -I_r , I_{d-r})$ on the right of $\tau_{t,U} - \tau_{t,I_d}$,
we see that
\begin{align}
&\| \tau_{t,U} - \tau_{t,I_d} \|_1 \nonumber \\
&\ge
\alpha \tr( CC^\dagger - AA^\dagger) + \beta(DD^\dagger - BB^\dagger) +(d-r)\beta - r \alpha \nonumber \\
&= 2(\alpha - \beta) \tr(CC^\dagger) \qquad \text{by Eq.~\eqref{eq:trCompU}} \nonumber \\
&=
\left(\frac{1+t}{r}- \frac{1-t}{d-r}\right) \tr C C^\dagger . \label{eq:dU-tau}
\end{align}
%%%%%%%%%%%%%%%%%%%%%%%%%%%%%%%%%%%%%%%%%%%%
\begin{lem}\label{lem:packingII}
Suppose $r = d/2$. Then,
there exists a finite subset $\{ U_i\} \subset \mathbb U(d)$ of cardinality $N \ge \exp(d^2/32)$
such that $\| \tau_{t,U_i} - \tau_{t,U_j} \|_1 > t/2$ for any $i \neq j$.
The Holevo $\chi_0$ fulfills $\chi_0 \le t^2$.
\end{lem}
\begin{proof}
Eq.~\eqref{eq:chi-tau} becomes $\chi /n \le \ln 2 - H( (1+t)/2 ) \le t^2$.
Eq.~\eqref{eq:dU-tau} says that if $\| \tau_{t,U} - \tau_{t,I} \|_1 \le t/2$,
then $(4/d) \tr CC^\dagger \le 1/2$.
Lemma~\ref{lem:projector-overlap} states that
this happens with probability at most $\exp( - d^2 / 32 )$.
The probabilistic existence argument applies.
\end{proof}

\begin{lem}\label{lem:packingIII}
Set $t=1$.
Suppose $\epsilon \in (0,1)$, and $r < d(1-\epsilon)/6$. Then,
there exists a finite subset $\{ U_i\} \subset \mathbb U(d)$ of cardinality $N \ge \exp( (1-\epsilon)r d / 2 )$
such that $\| \tau_{1,U_i} - \tau_{1,U_j} \|_1 > 2\epsilon$ for any $i \neq j$.
The Holevo $\chi_0$ fulfills $\chi_0 \le \ln (d/r)$.
\end{lem}
\begin{proof}
Eq.~\eqref{eq:chi-tau} becomes $\chi_0 \le \ln(d/r)$.
Eq.~\eqref{eq:dU-tau} says that if $\| \tau_{t,U} - \tau_{t,I} \|_1 \le 2\epsilon$,
then $\frac{d}{r^2} \tr A A^\dagger \ge (1-\epsilon)d/r$,
which is greater than $6$ when $r < d(1-\epsilon)/6$.
By Lemma~\ref{lem:projector-overlap}, this happens with probability
at most $\exp( - r^2  (1-\epsilon)d/2r ) = \exp(-rd(1-\epsilon)/2) $ .
The probabilistic existence argument applies.
\end{proof}

\subsection{Independent Measurement}
\label{sec:indLB}

\begin{proof}[Proof of Theorem~\ref{thm:LBprod} continued from p.~\pageref{incompletePf}]
Since $\sqrt{1-F}$ is a metric (Bures metric)
on the space of states, if there is a set of states $\rho_i$ such that
$1-F(\rho_i, \rho_j) > \delta$ for all $i \neq j$,
then for any $\rho$ there is at most one $\rho_i$ such that $1-F(\rho_i, \rho) \le \delta/4$.
In the regime where $\delta$ is close to 1,
we can use Packing Net III analyzed in Lemma~\ref{lem:packingIII}.
Since $1-T \ge 1- \sqrt{1-F^2} \ge F^2/2$,
we obtain a packing net of cardinality
$N = \exp( \Omega( rd (1-\delta)^4 ) )$
in which every state has rank at most $r$
and every pair has infidelity at least $\delta \in (0,1)$.

In order to compute Holevo information and to account for the small $\delta$ regime,
we consider the following set of states.
Define for $t \in (0,1)$ and $U \in \mathbb U(d-1)$
\begin{align}
\omega_{t,I} &=
\begin{pmatrix}
(1-t) &  & \\
 & t I_r/r & \\
 & & 0_{d-r-1}
\end{pmatrix} \\
\omega_{t,U} &= U \omega_{t,I} U^\dagger
\label{eq:deltaNet}
\end{align}
where $U$ is embedded into $\mathbb U(d)$ similarly as in Eq.~\eqref{eq:U-components}.
$\omega_{t,U}$ has rank $r+1 < d$.
Applying the defining formula $F = \tr \sqrt{ \sqrt{\omega_{t,I}} \omega_{t,U} \sqrt{\omega_{t,I}} }$
with the observation that $\omega_{t,U}$ is a mixture of two orthogonal states,
we obtain
\begin{align}
1-F(\omega_{t,U}, \omega_{t,I} ) = t ( 1-F(\tau'_{I}, \tau'_{U} ) )
\end{align}
where $\tau'_U = U \tau'_I U^\dagger $ is the $(d-1)$-dimensional state
that is maximally mixed on an $r$-dimensional subspace.
($\tau'_U$ is equal to $\tau_{t=1,U}$ of Eq.~\eqref{eq:tau} except the size.)
Since
\[
T^2 \le 2(1-F)
\]
by Eq.~\eqref{eq:FuchsGraaf},
we can apply the probabilistic existence argument
to find a set of states of cardinality $\exp( \Omega(rd) )$
that are $\delta=\Omega(t)$-separated in infidelity.

For the full rank case where the accuracy is measured in the trace distance,
we use Packing Net II analyzed in Lemma~\ref{lem:packingII},
from which we know there are $\exp( \Omega( d^2 ) )$ states
separated by the trace distance $\Omega(t)$.

Bounds for the Holevo information are supplied by the following two lemmas.
\end{proof}
\begin{lem}
Suppose $\vec M^{(a)}$ for each $a=1,\ldots,n$ is a POVM on $\mathbb C^d$.
Consider $\tau_{t,U}$ in Eq.~\eqref{eq:tau} with $r=d/2$.
For any distribution of unitaries $\{ U_j \} \subseteq \mathbb U(d)$
there exists $W \in \mathbb U(d)$
such that the Holevo information of
\[
\left\{  \tr \left( (\tau_{t,WU_j}) ^{\otimes n} \bigotimes_{a=1}^n \vec M^{(a)} \right) \right\}
\]
is at most $n t^2 / d$.
\label{lem:indLB-tracedistance}
\end{lem}
\begin{proof}
The first term of the Holevo information $\chi$
is the Shannon entropy of the distribution
\[
\mathfrak p = \E_\tau \tr\left(\tau^{\otimes n} \bigotimes_{a=1}^n \vec M^{(a)} \right)
\]
whose marginal is equal to $\E_\tau \tr( \tau \vec M^{(a)} )$.
By the subadditivity of entropy, we have
\[
H(\mathfrak p ) \leq \sum_{a=1}^n H( \E_\tau \tr( \tau \vec M^{(a)} ) ) .
\]
It follows that
\begin{align}
\chi \{ W \tau_j W^\dagger \}
\le \sum_{a=1}^n \chi_a \{ W \tau_j W^\dagger \}
\end{align}
where the subscript $a$ means with respect to $\vec M^{(a)}$.
Minimizing the right-hand side by varying $W$, we see there exists $W$
such that
\[
 \chi\{ \rho_{W U_j} \} \leq \min_V \sum_a \chi_a\{ \tau_{V U_j} \} .
\]
The minimum on the right-hand side is at most the average over $V$ from the Haar measure.
\[
\min_V \sum_a \chi_a\{ \tau_{V U_j} \} \le \E_V \sum_a \chi_a \{ \tau_{V U_j} \} .
\]
By concavity of entropy,
\[
\E_V \sum_a \chi_a\{ \tau_{V U_j} \} \le \sum_a \chi_a \{ \tau_{t,U}: \text{ Haar uniform }U \in \mathbb U(d) \} .
\]
Hence, it suffices to prove the lemmas when the initial distribution of $U_j$ is
Haar uniform, which we assume hereafter.
In addition, it suffices to consider rank-1 POVM elements
since one can always decompose a POVM element into rank-1 projectors of some positive weight.
Let each POVM element be $M_i = w_i d \proj{a_i}$.

The outcome probability is
\begin{align*}
p_i &\equiv \tr( M_i \tau_{t,U} ) = w_i d \left(\frac{2t}{d} \tr( P_1^{(i)} U P_{d/2} U^\dagger ) + \frac{1-t}{d} \right)\\
 &=: w_i (1-t+t Z_{d,d/2}^{(i)})
\end{align*}
where $P_1^{(i)} = \proj{a_i}$.
Since $\E_U Z_{d,d/2}^{(i)} = 1$ for any $i$ (see Eq.~\eqref{eq:Zmean} below),
the Holevo information per copy is
\begin{align*}
\chi_a
&= \sum_{i=1}^m - \E[ p_i] \ln \E[p_i] + \E_U [p_i \ln p_i] \\
&= \sum_{i=1}^m w_i \E_U\left[ (1-t+tZ_{d,d/2}^{(i)}) \ln (1-t+tZ^{(i)}_{d,d/2}) \right] \\
&\leq \sum_{i=1}^m w_i t^2 \left( \E_U \left[ (Z_{d,d/2}^{(i)})^2 \right] - 1 \right).
\end{align*}
Since $\E_U (Z^{(i)}_{d,d/2})^2 = (1+2/d)/(1+1/d)$ for any $i$ (see Eq.~\eqref{eq:Zvar} below),
we have
\begin{align}
\chi(\tau_{t,U}) \le \sum_{a=1}^n \chi_a \le n t^2/(d+1) .
\end{align}
This completes the proof of Lemma~\ref{lem:indLB-tracedistance}.
\end{proof}

\begin{proof}[Random variable $Z$]
Define the random variable $Z_{n,m}$ to be
\begin{equation}
Z_{n,m} := \frac{n}{m} \frac{x_1^2 + \cdots + x_{2m}^2}{x_1^2 + \cdots + x_{2n}^2}
\label{eq:Zdef}
\end{equation}
where $x_i$ are independent identical Gaussians with mean 0 and variance $1/2$.
Here we show
\begin{align}
\E Z_{n,m} &= 1, \label{eq:Zmean}\\
\E Z_{n,m}^2 &= \frac{1+ 1/m}{1+1/n}, \label{eq:Zvar}
\end{align}
by deriving the probability density function $p(Z_{n,m}=z)$ on $[0,n/m]$
\begin{align}
p(z) = \frac{m \Gamma(n)}{n \Gamma(n-m)\Gamma(m)}\left(\frac{mz}{n}\right)^{m-1}\left(1-\frac{mz}{n}\right)^{n-m-1}. \label{eq:Zdensity}
\end{align}
To this end, let $x = (x_1,\ldots, x_{2m})$ and $y = (x_{2m+1}, \ldots, x_{2n})$
be Cartesian coordinates for $\mathbb R^{2n}$.
Let $\rd^{2m-1} \Omega_x$ and $\rd^{2n-2m-1} \Omega_y$ be the solid angle elements
of respective dimensions.
Then the volume form $\rd V = \rd^{2m} x \rd^{2n-2m} y$ is equal to
$|x|^{2m-1}|y|^{2n-2m-1} \rd |x| \rd |y| \rd \Omega_x \rd \Omega_y$.
Defining new variables $r, \theta$ by $|x| = r \cos \theta$ and $|y| = r \sin \theta$
($\theta \in [0,\pi/2]$),
we see that the $(2n-1)$-dimensional solid angle element is
\[
\left.\frac{\rd V}{\rd r}\right\vert_{r=1} = \cos^{2m-1}\theta \sin^{2n-2m-1}\theta \rd \theta \rd \Omega_x \rd \Omega_y
\]
Since our variable $Z_{n,m} = (n/m) \cos^2 \theta =: (n/m)u$
is a function of $\theta$ only,
we integrate out $\rd \Omega_x \rd \Omega_y$,
and use the relation $\rd \theta = u^{-1/2}(1-u)^{-1/2} \rd u$
to arrive at Eq.~\eqref{eq:Zdensity} after normalization using Eq.~\eqref{eq:Dirichlet}.
\end{proof}

\begin{lem}
Let $t \in (0,1/3)$ and $d \ge 3$.
Suppose $\vec M^{(a)}$ for each $a=1,\ldots,n$ is a POVM on $\mathbb C^d$.
For any distribution of unitaries $\{ U_j \} \subseteq \mathbb U(d-1)$
there exists $W \in \mathbb U(d)$
such that the Holevo information of
\[
\left\{  \tr \left( (W\omega_{t,U_j}W^\dagger) ^{\otimes n} \bigotimes_{a=1}^n \vec M^{(a)} \right) \right\}
\]
is at most $4(n t^2/r) \ln (2/t) $ where $\omega$ is as in Eq.~\eqref{eq:deltaNet}.
\end{lem}
\begin{proof}
The first stage of the proof is similar to that of Lemma~\ref{lem:indLB-tracedistance};
we use the freedom $W$ and consider $W \omega_{t,VU_j} W^\dagger$
for some $W \in \mathbb U(d)$ and $V \in \mathbb U(d-1)$.
By varying $V$, we may assume that our ensemble $\mathcal M_W$ is
\[
\mathcal M_W = \{ W \omega_{t,U} W^\dagger : U \in \mathbb U(d-1) \text{ is Haar random. } \}
\]
and we will estimate the Holevo information per copy $\chi_a$ of this ensemble.

There is still remaining freedom to choose $\mathcal M_W$ using $W \in \mathbb U(d)$.
Certainly,
\[
\min_W \chi_a( \mathcal M_W ) \le \E_W \chi_a( \mathcal M_W )
\]
where the average over $W$ is with respect to Haar random $W$ and
the inequality is saturated when our POVM consists of
rank-1 projectors $\proj{v}$ from Haar uniform distribution,
which we assume hereafter.

The outcome probability density is
\begin{align*}
p(\ket{v},U) &\equiv \tr( d \proj{v} \omega_{t,U} )\\
 &= d(1-t) |v_1|^2 + \frac{t d}{r} (1-|v_1|^2)\tr(U P_r U^\dagger P_1)\\
 &= (1-t)\underbrace{Z_{d,1}}_{v} + t  \frac{d-Z_{d,1}}{d-1} \underbrace{Z'_{d-1,r}}_{z}
\end{align*}
where
$v_1$ is one component of the vector $v$,
$P_r$ and $P_1$ are $r$- and $1$-dimensional projectors, respectively,
and in the third line we used the notation in Eq.~\eqref{eq:Zdef}.  We
use the notation $Z_{d,1}, Z'_{d-1,r}$ to mean two independent random variables defined
according to \eq{Zdef} for appropriate choices of $n,m$.
Since we do not use all the degrees of freedom in $U$, we can think of
$v,z$ as our random variables (distributed according to $Z_{d,1},
Z'_{d-1,r}$ respectively), corresponding to outcome probability
\begin{align}
  p(v,z) = (1-t)v + t \frac{d-v}{d-1} z .
\end{align}
Now,
\begin{align*}
\chi_a = \E_v
\left[
- p(v,\E_z[z]) \ln p(v,\E_z[z]) + \E_z [p(v,z) \ln p(v,z)]
\right],
\end{align*}
and we use $\E_z z = 1, \E_z z^2 = (1+1/r)/(1+1/(d-1))$, and
\[
\ln p(v,z) \le \ln p(v,1) + \frac{p(v,z)-p(v,1)}{p(v,1)},
\]
to obtain
\begin{align*}
\chi_a
&\le
\frac{t^2(d-r-1)}{r d (d-1)} \cdot \E_v \frac{(d-v)^2}{td + (d - td -1) v}\\
&\le
\frac{t^2}{r} \cdot \E_v \frac{1}{t+v/3}
\end{align*}
where the last line is because $t < 1/3$ and $d \ge 3$.

The probability density function of $v=Z_{d,1}$ is given by Eq.~\eqref{eq:Zdensity}
\[
f(v) = (1 - 1/d)(1-v/d)^{d-2} < 1.
\]
Since $t > 0$,
\begin{align*}
\E_v \frac{1}{t+v/3}
&= \left( \int_0^3 + \int_3^d \right) \frac{f(v)}{t+v/3} \rd v \\
&\le \int_0^3 \frac{\rd v}{t+v/3} + \int_3^d f(v) \rd v \\
&< 3 \ln(1+1/t) + 1 \\
& < 4 \ln(2/t).
\end{align*}
This completes the proof.
\end{proof}

%%%%%%%%%%%%%%%%%%%%%%%%%%%%%%%%%%%%%%%%%%%%
\section{Implementation on a quantum computer}
%%%%%%%%%%%%%%%%%%%%%%%%%%%%%%%%%%%%%%%%%%%%
In this section we informally describe how our tomography strategy can be implemented
in time $n^{O(dr)}$ on a quantum computer.

Our measurement involves a POVM with a continuously infinite number of
outcomes.  However, it can be approximated with a finite POVM using
ideas from \cite{Winter:02a}.  The first step is to measure $\lambda$,
as proposed by Keyl-Werner~\cite{Keyl01}.  This can be done
efficiently using the Schur transform~\cite{BCH05a} or the quantum
Fourier transform over the symmetric group~\cite{Beals97,Har05}.

Next, we would like to find a collection of unitaries $U_1,\ldots,U_m$
such that
\[
\frac 1 m \sum_{i=1}^m M(\lambda, U_i) \approx \Pi_\lambda.
\]
This can be done by choosing $m = \tilde O(\dim \cQ_\lambda / \eps^2)$ random
unitaries, as proven in \cite{Winter:02a}, which in turn was based on
\cite{AW02}).
% The operator Chernoff bound has a dimension factor in front of exp.
The resulting measurement can be implemented by the isometry
\[
V = m^{-1/2} \sum_{i=1}^m \sqrt{M(\lambda,U_i)} \otimes \ket{i}.
\]
Using the Schur transform, this reduces to performing the isometry
\[
\tilde V = C \sum_{i=1}^m \sqrt{\bq_\lambda(U_i\bar\lambda U_i^\dag)} \otimes \ket i,
\]
where $C$ is a normalizing constant.  This isometry can be implemented
using $O((\dim \cQ_\lambda)^2 m^2)$ gates~\cite{isometries}, which
is $\tilde O(n^{2dr}/\eps^2)$.

We conjecture that run-time $\poly(n,d,\ln(1/\eps))$ is possible, but
do not know how to achieve this, even in the relatively simple case of
$r=1$.

\section{Discussion}

The sample complexity of the general quantum tomography problem
is nearly resolved here.
It is confirmed up to logarithmic factors that
one only needs as many copies as the number of unknown parameters
if one can perform joint measurements.
In addition, we have shown information-theoretically
that this optimal measurement {\em cannot} be a combination
of independent measurements.
Our result raises an important question
on the performance of {\em adaptive} measurements,
where an individual copy is measured at a time,
but each measurement may utilize the history of outcomes on other copies.
Is there an asymptotic separation between the power of
adaptive and collective measurements?
Another open problem is whether our joint measurement scheme
can be implemented efficiently on a quantum computer;
we briefly remark that the implementation is possible
in a polynomial time in $n$ for a fixed $d$,
but dependence on $d$ is exponential.
There is a method to extract the eigenvectors of a small-rank density matrix
on a quantum computer efficiently~\cite{LloydMohseniRebentrost2014PCA},
but it remains challenging to convert the eigenvector
into a classical description.

An independent and concurrent work~\cite{OW-tomo}
analyzes Keyl's measurement strategy~\cite{Keyl06},
and proves that it only requires $n = O(dr/\eps^2)$ copies
to achieve $\epsilon$ accuracy in trace distance.
This improves on our corollary for trace distance by removing the logarithmic factor,
but does not imply our fidelity bound, which is incomparable to theirs.

\begin{acknowledgments}
We thank
Robin Blume-Kohout, Steve Flammia, Masahito Hayashi, Debbie Leung, and John Watrous
for discussions.
We also thank
Ryan O'Donnell and John Wright for sharing their draft of \cite{OW-tomo} with us.
JH is supported by the Pappalardo Fellowship in Physics while at MIT.
AWH was funded by NSF grants CCF-1111382 and
CCF-1452616 and ARO contract W911NF-12-1-0486.
ZJ and NY's research was supported by NSERC, NSERC DAS, CRC, and CIFAR.
XW's research was funded by ARO contract W911NF-12-1-0486
and by the NSF Waterman Award of Scott Aaronson.
Part of the research was conducted when XW was visiting Institute for Quantum Computing (IQC),
University of Waterloo and XW thanks IQC for its hospitality.
\end{acknowledgments}

\appendix
\section{Overlap of random projectors}
\label{sec:randomProjectorOverlap}

Here, we provide a self-contained proof of Lemma~\ref{lem:projector-overlap}
(Lemma III.5 of Ref.~\cite{HaydenLeungWinter2006}).
We follow the ideas of Ref.~\cite{HaydenLeungWinter2006} and \cite{HaydenLeungShorWinter2004}.

\begin{lem}
Let $\mathcal D$ be the set of all $d \times d$ normalized density matrices of rank $p$,
and $\Delta$ be the set of all probability vectors $\eta$ of length $p$.
Suppose $\mathcal D$ has a $\mathbb U(d)$-invariant probability measure $\rd \rho$.
Then, there exists a permutation-symmetric probability measure
$\rd \eta$ on $\Delta$ such that
\[
\iint \rd \eta \ \rd U ~f( U \eta U^\dagger ) = \int \rd \rho ~f( \rho )
\]
for any continuous function $f$ on $\mathcal D$
where $\rd U$ is the normalized Haar measure on $\mathbb U(d)$,
and $\eta$ in between $U$ and $U^\dagger$ denotes the diagonal matrix with entries
$(\eta_1,\ldots, \eta_p, 0,\ldots,0)$.
\label{lem:rv-decompose}
\end{lem}
This means that the eigenvalues and the eigenvectors
can be treated as if they were ``independent random variables.''
Strictly speaking, $\rd \eta$ and $\rd U$ are {\em not} derived from $\rho$;
we just find that they induce the measure $\rd \rho$ on $\mathcal D$
by the map $(\eta,U) \mapsto U \eta U^\dagger$.
\begin{proof}
Since {\em sorted} eigenvalues are continuous functions of the matrix,
we have a map $\lambda: \mathcal D \to \Delta^\downarrow$,
which induces a measure $\rd \lambda$ on $\Delta^\downarrow$,
the set of all sorted non-negative $p$ real numbers summing to 1.
The defining equation for the induced measure is
$\int \rd \rho~ g( \lambda(\rho) ) = \int \rd \lambda ~g(\lambda)$
for any continuous function $g$.
Here, we have identified a vector with a diagonal matrix padded with $(d-p)$ zeros.
Define
\[
\bar f(\rho) = \int \rd U f( U \rho U^\dagger )
\]
so that $\bar f(\rho) = \bar f( V \rho V^\dagger )$ for any $V \in \mathbb U(d)$.
Since $\rd \rho$ is unitary invariant,
$\int \rd \rho f(\rho) = \int \rd \rho f(U\rho U^\dagger)$.
Integrating the both sides over $U$,
$\int \rd \rho f(\rho) = \int \rd U \int \rd \rho f( U \rho U^\dagger) = \int \rd \rho \bar f( \rho )$.
(All spaces are compact, so integration order never matters.)
We can now prove an analogous version of the lemma for $\Delta^\downarrow$:
\begin{align*}
\int \rd \rho f(\rho)
&= \int \rd \rho \bar f( \rho ) = \int \rd \rho \bar f( \lambda(\rho) ) \\
&= \int \rd \lambda \bar f( \lambda ) = \iint \rd \lambda \rd U f ( U \lambda U^\dagger )
\end{align*}
In order to finish the proof, all we need is to divide $\Delta$ into $p!$ pieces,
each of which is mapped to $\Delta^\downarrow$ by permuting components
up to measure zero sets,
and assign measure to each piece by $\rd \lambda / p!$.
Thus defined $\rd \eta$ on $\Delta$ is permutation-invariant.
\end{proof}

\begin{lem}
Let $x_1,x_2,\ldots$ be independent gaussian random variables with mean $0$ and variance $\frac12$.
Let $U$ be a Haar random unitary of dimension $d$,
and $P$ and $Q$ be $d$-dimensional projectors of rank $p$ and $q$, respectively.
For any real number $\xi$, it holds that
\[
\mathbb E_{x_i} \exp\left[ \xi \sum_{i=1}^{2pq} x_i^2 \right]
\ge
\mathbb E_U \exp\left[ \xi d \tr(QUPU^\dagger) \right].
\]
\end{lem}
\begin{proof}
Consider $\mathbb C^{dp} = \mathbb C^d \otimes \mathbb C^p$,
and define $Q' = Q \otimes I_p$ to be the projector of rank $qp$.
Without loss of generality, we assume that $P, Q$ are diagonal.
The random tuple $(x_1,\ldots,x_{2dp})$ has the probability density $\frac{1}{\pi^{dp}}\exp( - r^2 ) \rd^{2dp} x$
where $r^2 = \sum_{i=1}^{2dp} x_i^2$.
This means in particular that the magnitude variable $r$
and the direction variable $\hat x = (x_1,\ldots,x_{2dp}) / r$ are independent.
The direction variable $\hat x$ defines a normalized pure state $\ket{ \hat x}$ on $\mathbb C^d \otimes \mathbb C^p$,
and the sum $\sum_{i=1}^{2pq} \hat x_i^2$ can be regarded as the squared norm of $Q' \ket{\hat x}$.
\begin{align*}
\sum_{i=1}^{2pq} x_i^2 = r^2 \bra{\hat x} Q' \ket{\hat x}  = r^2 \tr Q \rho
\end{align*}
where $\rho$ is the reduced density matrix of $\ket{\hat x}$ on $\mathbb C^d$.

As a random variable, $\rho$ defines a $\mathbb U(d)$-invariant measure on
the set of all density operators of rank at most $p$.
By Lemma~\ref{lem:rv-decompose}, $\rho$ may be replaced with
a random vector variable $\eta$ and a Haar random $U$.
Due to the permutation invariance and the normalization,
we have $\mathbb E \eta_i = \mathbb E \eta_{j} = 1/p$,
so $\mathbb E_\eta \sum_i \eta_i \ket i \bra i = P /p$.

By the convexity of $\exp$,
\begin{align*}
\mathbb E_{x_i}& \exp\left[ \xi \sum_{i=1}^{2pq} x_i^2 \right] \\
&=
\mathbb E_r \mathbb E_\eta \mathbb E_U
\exp \left[ \xi r^2 \tr Q U \eta U^\dagger \right] \\
&\ge
\mathbb E_U \exp \left[
 \xi (\mathbb E_r r^2 ) \mathbb E_\eta \tr Q U \eta U^\dagger
\right] \\
&=
\mathbb E_U \exp \left[
 \xi (d p ) \tr QU (P/p) U^\dagger
\right].
\end{align*}
we complete the proof.
\end{proof}

\begin{proof}[Proof of Lemma~\ref{lem:projector-overlap}]
Recall Markov's inequality: For non-negative real random variable $X$ and $a > 0$, $\Pr[ X \ge a ] \le \mathbb E X / a$.
This is easily seen once we define $Y = a$ if $X \ge a$ and $Y=0$ if $X < a$, so $Y \le X$.
Then, $\Pr[ X \ge a] = \Pr[ Y = a] = \mathbb E Y / a \le \mathbb E X / a$.

Let us abbreviate $\frac{d}{pq} \tr QUPU^\dagger$ as  $Z$.
For any $\xi > 0$ and $z > 0$,
\begin{align*}
\Pr[ Z \ge 1 + z ]
&= \Pr[ e^{\xi Z} \ge e^{\xi(1+z)} ] \\
&\le \mathbb E_U e^{\xi Z} e^{-\xi(1+z)} \\
&\le
\mathbb E_{x_i} \exp\left[ \frac{\xi}{pq} \sum_{i=1}^{2pq} x_i^2 \right] e^{-\xi(1+z)} \\
&= e^{-\xi(1+z)} \left( 1 - \frac{\xi}{pq} \right)^{- pq}
\end{align*}
The last equality is directly evaluated with PDF $\frac{1}{\sqrt{\pi}} e^{-z^2}$.
The best bound is when $ \xi = pq z/ (1+z) > 0$.
Substituting this value for $\xi$, we prove the first inequality in the theorem.

The opposite direction goes similarly.
Let $\xi > 0$ and $z \in (0,1)$.
\begin{align*}
\Pr[ Z \le 1 - z ]
&= \Pr[ e^{-\xi Z} \ge e^{-\xi(1-z)} ] \\
&\le \mathbb E e^{-\xi Z} e^{\xi(1-z)} \\
&\le
\mathbb E \exp\left[ -\frac{\xi}{pq} \sum_{i=1}^{2pq} x_i^2 \right] e^{\xi(1-z)} \\
&= e^{\xi(1-z)} \left( 1 + \frac{\xi}{pq} \right)^{- pq}
\end{align*}
The best bound is when $\xi = pq z /(1-z) > 0$.
Substituting this value for $\xi$, we prove the second inequality in the theorem.

The last inequality can be proved by examining extreme values of,
for example, $g(z) = z - \ln(1+z) - (1-\ln 2) z^2$.
The minimum values in the range $z \in (-1,1]$ occur at $z = 0, 1$,
where $g(z) = 0$.
\end{proof}

\bibliography{opt-tomo}
\end{document}